\documentclass[10pt, twocolumn, comsoc]{IEEEtran}

\usepackage{graphicx,epsfig}
\usepackage[noadjust]{cite}
\usepackage{mcite}
\usepackage{amsfonts,helvet}
\usepackage{fancyhdr}
\usepackage{threeparttable}
\usepackage{epsf,epsfig}
\usepackage{amsthm}
\usepackage{amsmath}
\usepackage{siunitx}
\usepackage{amssymb}

\usepackage{dsfont}
\usepackage{subfigure}
\usepackage{color}
\usepackage[linesnumbered,noend,ruled]{algorithm2e}
\usepackage{algpseudocode}
\usepackage{algcompatible}
\usepackage{enumerate}
\usepackage{cancel}
\usepackage{multirow}
\usepackage{bbm}
\usepackage{graphicx,subfigure}
\usepackage{graphicx}

\newtheorem{lemma}{Lemma}

\usepackage{eucal}

\def\bff{{\bf f}}
\def\bg{{\bf g}}
\def\bh{{\bf h}}

\def\bR{{\bf R}}

\def\cC{\mbox{$\mathcal{C}$}}

\def\cN{\mbox{$\mathcal{N}$}}

\def\bbC{\mbox{$\mathbb{C}$}}

\def\bbE{\mbox{$\mathbb{E}$}}

\usepackage{graphicx,epsfig}
\usepackage[noadjust]{cite}
\usepackage{mcite}
\usepackage{amsfonts,helvet}
\usepackage{fancyhdr}
\usepackage{threeparttable}
\usepackage{epsf,epsfig}
\usepackage{amsthm}
\usepackage{amsmath}
\usepackage{siunitx}
\usepackage{amssymb}

\usepackage{dsfont}
\usepackage{subfigure}
\usepackage{color}
\usepackage[linesnumbered,ruled]{algorithm2e}
\usepackage{algpseudocode}
\usepackage{algcompatible}
\usepackage{enumerate}
\usepackage{cancel}
\usepackage{bbm}
\usepackage{graphicx,subfigure}
\usepackage{graphicx}

\usepackage{eucal}

\makeatletter
\def\blfootnote{\xdef\@thefnmark{}\@footnotetext}
\makeatother

\def\bff{{\bf f}}
\def\bg{{\bf g}}
\def\bh{{\bf h}}

\def\bR{{\bf R}}

\def\cC{\mbox{$\mathcal{C}$}}

\def\cN{\mbox{$\mathcal{N}$}}

\def\bbC{\mbox{$\mathbb{C}$}}

\def\bbE{\mbox{$\mathbb{E}$}}

\title{A Selective Secure Precoding Framework for MU-MIMO Rate-Splitting Multiple Access Networks Under Limited CSIT}

\author{Sangmin~Lee, \textit{ Student Member, IEEE},  Seokjun~Park, \textit{ Student Member, IEEE}, Jeonghun~Park, \textit{Member, IEEE}, and Jinseok~Choi, \textit{Member, IEEE}

\thanks{S. Lee is with the Department of Electrical Engineering, Ulsan National Institute of Science and Technology (UNIST), Ulsan, 44919, South Korea (e-mail: {\texttt{sangminlee@unist.ac.kr}}).

S. Park and J. Choi are with the School of Electrical Engineering, Korea Advanced Institute of Science and Technology (KAIST), Daejeon, 34141, South Korea (e-mail: {\texttt{\{sj.park, jinseok\}@kaist.ac.kr}}).

J. Park is with the School of Electrical and Electronic Engineering, Yonsei University, Seoul, 03722, South Korea (e-mail: {\texttt{jhpark@yonsei.ac.kr}}).

}}

\begin{document}

\maketitle \setcounter{page}{1} 
\begin{abstract}
    In this paper, we propose a robust and adaptable secure precoding framework designed to encapsulate a intricate scenario where legitimate users have different information security: secure private or normal public information.
    Leveraging rate-splitting multiple access (RSMA),
    we formulate the sum secrecy spectral efficiency (SE) maximization problem in downlink multi-user multiple-input multiple-output (MIMO) systems with multi-eavesdropper.
    To resolve the challenges including the heterogeneity of security, non-convexity, and non-smoothness of the  problem, we initially approximate the  problem  using a LogSumExp technique.
    Subsequently, we derive the first-order optimality condition in the form of a generalized eigenvalue problem.
    We utilize a power iteration-based method to solve the condition, thereby achieving  a superior local optimal solution.
    The proposed algorithm is further extended to a more realistic scenario involving limited channel state information at the transmitter (CSIT).
    To effectively utilize the limited channel information, we employ a conditional average rate approach.
    Handling the conditional average by deriving useful bounds, we establish a lower bound for the objective function under the conditional average.
    Then we apply the similar optimization method as for the perfect CSIT case.
    In simulations, we validate the proposed algorithm in terms of the sum secrecy SE.
\end{abstract}
\renewcommand\IEEEkeywordsname{Index Terms}
\blfootnote{This work was presented in part at the {\em IEEE Vehicular Technology Conference (IEEE VTC) Workshop}, Florence, Italy, 2023 \cite{lee2023rate}.}

\section{Introduction} \label{sec:intro}

The shift from the era of 5G to 6G marks a significant leap in the technological evolution of wireless communication, requiring careful consideration of various complex factors, not just increased transmission capacity \cite{yang20196g,zhang20196g}.
With the increasing complexity of wireless communication networks, the current era has seen a notable rise in the risks related to unauthorized access, eavesdropping, and data breaches \cite{zou2016survey}.
This highlights the urgent need for new solutions that not only meet the growing demands for faster data transmission but also strengthen the security of wireless systems.
In this regard, rate-splitting multiple access (RSMA) can be adopted to increase spectral efficiency (SE) by reducing inter-user interference \cite{joudeh2016robust}.
Concurrently, physical layer security techniques are used to enhance the security measures essential for safeguarding communications \cite{sulyman2022physical}.
Together, these strategies play a key role in building a strong framework that ensures both efficient and secure wireless communications.

\subsection{Prior Work}

Improving SE is a key focus in wireless communications research, driven by the growing need for more data bandwidth and better connectivity.
To manage the inter-user interference, the non-orthogonal multiple access (NOMA) methodology stands out for its effectiveness \cite{shin2017non}.
Moreover, RSMA as represented in \cite{clerckx2016rate}, has shown promise in effectively handling inter-user interference.
RSMA systems with low-resolution quantizers have shown to be a versatile and promising technology for future 6G networks \cite{park2023rate}.
More importantly, RSMA has demonstrated its superiority over conventional multiple access methods by bridging spatial-division multiple access (SDMA), orthogonal multiple access (OMA), and NOMA \cite{mao2018rate}.
RSMA was highlighted in \cite{joudeh2016sum} for its theoretical superiority in maximizing channel degree-of-freedom under imperfect channel state information (CSI) by mitigating interference.
A principle of RSMA involves the division of each user's transmission into two parts: the common and private messages \cite{clerckx2016rate}.
The common stream is decodable by all users through successive interference cancellation (SIC) when decoding private streams, thereby reducing interference and enhancing SE \cite{mao2018rate}.

The advantages of RMSA have drawn substantial research endeavors on RSMA beamforming design.
The potential of RSMA for boosting sum SE in multiple-input single-output (MISO)  channels was examined by using a  precoding strategy \cite{hao2015rate}.
To further elevate the SE, a linear precoding approach for the downlink MISO  system predicated on weighted minimum mean square error (WMMSE) was proposed in \cite{joudeh2016sum}.
Additionally, the design of RSMA precoders was articulated by casting the optimization problem into a convex form \cite{li2020rate}.
The optimization of rate allocation and power control within RSMA which aimed at maximizing the sum rate was explored in \cite{yang2021optimization}.
A hierarchical RSMA architecture, facilitating the decoding of more than one common stream based on hierarchy, was introduced for downlink massive multiple-input multiple-output (MIMO) systems \cite{dai2016rate}.
An RSMA precoder design algorithm, leveraging the generalized power iteration (GPI) method for sum SE maximization in the downlink multi-user MIMO (MU-MIMO) system considering channel estimation errors, was developed in \cite{park2022rate}.

In addition to the SE performance, it has been important  to guarantee the  information security.
For instance, as an exemplar among 6G applications, Internet-of-Things (IoT) systems encounter a pronounced vulnerability to the pervasive issue of wiretapping \cite{haider2020optimization}.
While cryptography was initially introduced as a method to ensure secure communication \cite{massey1988introduction}, its applicability encounters limitations in physically constrained environments, such as those prevalent in IoT scenarios.
In this regard, physical layer security has been receiving great attention thanks to its low computational complexity \cite{wyner1975wire}.
Starting with the proposal of Wyner \cite{wyner1975wire}, physical layer security has received constant attention, and many studies have been conducted \cite{{shannon1949communication}, {shiu2011physical}}.
Physical layer security stands out by ensuring a positive transmission rate for authorized or legitimate users \cite{gopala2008secrecy,khisti2008secure}.
Building on prior research, several precoding schemes have been proposed to enhance the secrecy SE.
The cases of a single eavesdropper were explored to evaluate the achievable secrecy rate in MISO channels \cite{shafiee2007achievable,5485016}.
Additionally, the scenario of multiple antennas eavesdropper was introduced in the context of MIMO system \cite{khisti2010secure}, and a colluding-eavesdropper scenario was considered in \cite{choi2021sum}.
Security in device-to-device systems was also studied \cite{choi2020securelinq}.
In more complex systems, the study in \cite{kampeas2016secrecy} investigated the secrecy rate in scenarios where multiple users and multiple eavesdroppers coexist within the network.
For a multi-user multi-eavesdropper network, the use of artificial noise (AN)-aided secure precoding was proposed  in \cite{wang2016secrecy, choi2022joint} to enhance secure communication.

Recently, there has been collaborative research focusing on the development of secure wireless communication techniques considering the RSMA approach.
The RSMA-based algorithm introduced in \cite{xia2022secure} considered the secure beamforming optimization for MISO networks under both perfect and imperfect channel state information at the transmitter (CSIT).
In \cite{fu2020robust}, the secure RSMA scheme was suggested in a two-user MISO system
with imperfect CSIT.
A cooperative rate-splitting technique to maximize the sum secrecy rate in the MISO broadcast channel system was suggested specifically for two legitimate users and one eavesdropper \cite{li2020cooperative}.
A novel RSMA-based secure transmission scheme with AN was further proposed in \cite{cai2021resource}.

Although many studies have offered valuable contributions in beamforming optimization for physical layer security, an adaptive secure beamforming for a comprehensive network scenario has not been tackled yet.
For instance, the works in \cite{xia2022secure,fu2020robust} did not assume any eavesdropper, only considering legitimate users as potential eavesdroppers.
In addition, the considered systems in \cite{li2020cooperative, cai2021resource} were restricted to a single eavesdropper.
Besides the existence of the eavesdropper, there can be  different conditions for required security among legitimate users, depending on the type of messages.
This needs a selective secure precoding approach to guarantee an optimal communication performance.
Accordingly, it is important to consider both legitimate users with different security requirements and explicit eavesdroppers for optimizing secure beamforming.

When considering such a comprehensive network, 
it is essential to account for limitations of CSIT: imperfect CSIT for legitimate users and partial CSIT for eavesdroppers.
In this regard, the fore-mentioned works \cite{li2020cooperative, cai2021resource} developed secure beamforming methods considering a perfect CSIT scenario, which implies the necessity of investigating RSMA secure beamforming under a limited CSIT scenario.
Therefore, it is desirable to develop an optimization framework encompassing not only  a general scenario of multi-user and multi-eavesdropper MIMO systems with different security priorities but also a limited CSIT assumption.
\begin{table*}[t]
    \caption{Comparative Analysis with Prior Research Endeavors}
    \label{tab}
    \centering
    {
    \begin{tabular}{|c|cc|cc|cc|c|c|}
    \hline \hline
    \multirow{2}{*}{\bf{Prior Work}} & \multicolumn{2}{c|}{\bf{Multiple}} & \multicolumn{2}{c|}{\bf{Partial CSIT}} & \multicolumn{2}{c|}{\bf{Potential Wiretappers}} & \multirow{2}{*}{\bf{AN}} & \multirow{2}{*}{\bf{RSMA}} \\ \cline{2-7}
     & \multicolumn{1}{c|}{\bf{Users}} & \bf{Eavesdroppers} & \multicolumn{1}{c|}{\bf{Users}} & \bf{Eavesdroppers} & \multicolumn{1}{c|}{\bf{Users}} & \bf{Eavesdroppers} &  &  \\ \hline
    Sum Rate Maximization WMMSE \cite{joudeh2016sum} & \multicolumn{1}{c|}{\checkmark} & N/A & \multicolumn{1}{c|}{\checkmark} & N/A & \multicolumn{1}{c|}{} & N/A &  & \checkmark \\ \hline
    Max-Min Fairness GPI \cite{lee2023max} & \multicolumn{1}{c|}{\checkmark} & \checkmark & \multicolumn{1}{c|}{\checkmark} & \checkmark & \multicolumn{1}{c|}{} & \checkmark &  &  \\ \hline
    Max-Min Fairness SCA \cite{fu2020robust} & \multicolumn{1}{c|}{\checkmark} & N/A & \multicolumn{1}{c|}{\checkmark} & N/A & \multicolumn{1}{c|}{\checkmark} & N/A &  & \checkmark \\ \hline
    Sum Secrecy Rate Maximization GPI \cite{choi2022joint} & \multicolumn{1}{c|}{\checkmark} & \checkmark & \multicolumn{1}{c|}{} & \checkmark & \multicolumn{1}{c|}{} & \checkmark & \checkmark &  \\ \hline
    Sum Secrecy Rate Maximization SCA \cite{xia2022secure} & \multicolumn{1}{c|}{\checkmark} & N/A & \multicolumn{1}{c|}{\checkmark} & N/A & \multicolumn{1}{c|}{\checkmark} & N/A &  & \checkmark \\ \hline
    Energy Efficiency SWIPT SCA \cite{lu2021worst} & \multicolumn{1}{c|}{\checkmark} & N/A & \multicolumn{1}{c|}{\checkmark} & N/A & \multicolumn{1}{c|}{\checkmark} & N/A &  & \checkmark \\ \hline
    Proposed & \multicolumn{1}{c|}{\checkmark} & \checkmark & \multicolumn{1}{c|}{\checkmark} & \checkmark & \multicolumn{1}{c|}{\checkmark} & \checkmark &  & \checkmark \\ \hline
    \end{tabular}
    }   
\end{table*}

\subsection{Contributions}
This paper presents a secure precoding optimization framework with selectively applicable physical layer security for MIMO-RSMA systems where multiple users and eavesdroppers coexist under limited CSIT.
Our primary contributions are summarized as follows:

\begin{itemize}
\item
To present a comprehensive precoding optimization framework concerning physical layer security in RSMA systems, we consider three distinct groups in the network: the secret user group which requires  information security in their private stream; the normal user group which operates without security concerns; and the eavesdroppers who are not a legitimate users and  wiretap secret users' private stream.
Each secret user assumes all the other users to be potential eavesdroppers as well as the actual eavesdroppers.
In the considered system, our focus lies in maximizing the sum secrecy SE while enforcing secure precoding for the secret user group and regular precoding for the normal user group.


\item 
We first propose a selective secure precoding framework to maximize the sum secrecy SE under perfect CSIT, addressing several key challenges such as the non-convexity and non-smoothness.
In particular, the common rate of RSMA envolves a min function for decodability, and the secrecy SE is determined by selecting the maximum wiretap SE among eavesdroppers.
To overcome these issues, we approximate the non-smooth objective function to obtain a smooth sum secrecy SE expression.
Subsequently, we reformulate the problem into a more tractable non-convex form, expanding it into a higher-dimensional space.
We derive the first-order optimality condition, interpreting it as a generalized eigenvalue problem, and employ a power iteration-based algorithm to find the superior stationary point.

\item 
We extend the proposed algorithm to the case of limited CSIT for both the legitimate users and eavesdroppers.
To take advantage of the limited CSIT, we adopt the concept of the ergodic SEs.
Given that only the estimated legitimate channel and long-term channel statistics are accessible at the access point (AP), we then employ a conditional averaged secrecy SE approach and derive a lower bound of the objective function as a function of the available channel knowledge.
For the derived objective function, we apply a similar optimization framework as proposed in the case of perfect CSIT.
Consequently, fully exploiting the available channel information through the reformulation and lower bounding, the proposed algorithm becomes more robust to the limited CSIT.
Table~\ref{tab} summarizes our contributions comparing with the state-of-the-art methods.
Based on the comparison, our algorithm demonstrates greater generality, making it highly adaptable to a wide range of physical layer security systems, surpassing the flexibility of existing methods.



\item 
Via simulations, we validate the performance of the proposed robust and selective secure precoding method within RSMA networks.
By adjusting the three groups, the algorithm demonstrates improved generalization, resulting in superior performance compared to the benchmark methods.
In particular, it is observed that accommodating different security priorities among users can achieve a significant gain in the sum secrecy SE.
Therefore, we conclude that the algorithm effectively provides a comprehensive framework that accommodates a wide range of scenarios in multi-user and multi-eavesdropper MIMO systems under both perfect and limited CSIT.

\end{itemize}

{\it Notation}: $\bf{A}$ is a matrix and $\bf{a}$ is a column vector. 
$(\cdot)^{\sf T}$, $(\cdot)^{\sf H}$, and $(\cdot)^{-1}$ denote transpose, Hermitian, and matrix inversion, respectively.
$\bbC$ denotes the complex domain.
${\bf{I}}_N$ is the identity matrix with size $N \times N$, and $\bf 0$ represents a matrix or column vector with a proper size.
For ${\bf{A}}_1, \dots, {\bf{A}}_N \in \mathbb{C}^{K \times K}$, ${\bf{A}} = {\rm blkdiag}\left({\bf{A}}_1,\dots, {\bf{A}}_N \right) \in \bbC^{KN\times KN}$ is a block diagonal matrix. $\|\bf A\|$ represents L2 norm.
$\mathbb{E}[\cdot]$, ${\rm Tr}(\cdot)$, and $\otimes$ represent  expectation operator, trace operator, and Kronecker product, respectively.
The vector ${\bf{e}}_i$ is defined as a column vector with a single 1 element at the $i$th position and all other elements set to 0.
We use ${\rm{vec}}(\cdot)$ for vectorization.
$\cC\cN(\mu,\sigma^2)$ indicates a complex Gaussian distribution with mean $\mu$ and variance $\sigma^2$.
$[\cdot]^+$ represents $\max[\cdot, 0]$.

\section{System Model} \label{sec:sys_model}

We consider a single-cell downlink MU-MIMO system where an AP equipped with $N$ antennas serves $K$ users with a single antenna.
In the considered system, we assume the existence of two distinct legitimate user classes: one is the secret user group whose information security needs to be guaranteed, and the other is the normal user group whose information security is not required, e.g., public information.
Our system also accommodates multiple 
 eavesdroppers, each equipped with a single antenna.
An example scenario of our system model is shown in Fig.~\ref{figure_1}.
Let us define the total user set, secret user set, normal user set, and eavesdropper set as $\CMcal{K} = \{1, ... ,K\}$, $\CMcal{S} = \{1, ... ,S\}$, $\CMcal{M} = \{S+1, ... ,S+M\}$, and $\CMcal{E} = \{1, ... ,E\}$, respectively.
We assume  $K = S + M$.

\subsection{Signal Model}
\begin{figure}[!t]
    \centerline{\includegraphics[width=\columnwidth]{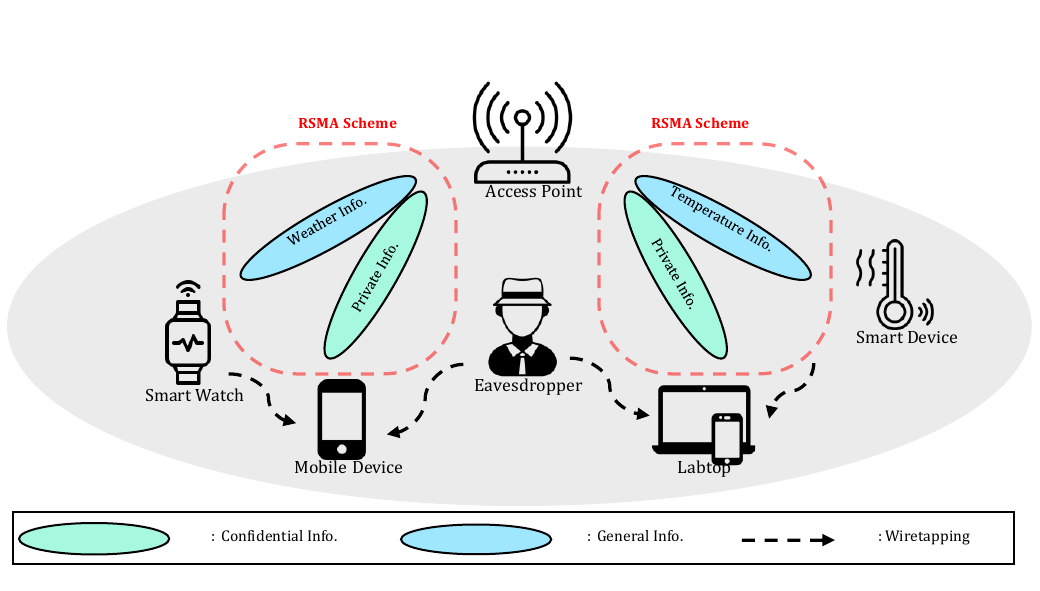}}
    \caption{An example scenario of a multiuser MIMO downlink system employing RSMA with multiple eavesdroppers.
    }
    \label{figure_1}
\end{figure}

We consider the single-layer RSMA method \cite{clerckx2016rate}.
The RSMA-encoded message consists of the common stream $s_{\sf c}$ and private stream $s_{k}$, $\forall k\in\CMcal{K}$.
We consider $s_{\sf c} \sim \mathcal{CN}(0,{P})$ and $s_{k} \sim \mathcal{CN}(0,{P})$.
Our approach entails the design of the common stream $s_{\sf c}$ to ensure decoding feasibility for all $K$ users, while the private stream $s_{s}$, ${ s}\in\CMcal{S}$ intended for secret users is crafted with a focus on information security considerations.
As established earlier, the private stream $s_{m}$, $m\in\CMcal{M}$ designated for normal users operates without necessitating security considerations.
We operate under the assumption that only the private stream of the secret user carries critical information, while the common stream is comprised of less sensitive data that does not necessitate stringent security measures as discussed in \cite{xia2022secure}.
Consequently, our objective revolves around safeguarding the confidentiality of the private streams belonging to the secret users by preventing any information leakage.

We represent the transmit signal ${\bf{x}} \in \mathbb{C}^{N}$ using linear precoding as follows: 
\begin{align} \label{eq:transmit_signal}
    {\bf{x}} = {\bf{F}}{\bf{s}} = {\bf{f}}_{\sf c} s_{\sf c} + \sum_{i = 1}^{K} {\bf{f}}_i s_i,
\end{align}
where ${\bf{F}}=[{\bf{f}}_{\sf{c}},{\bf{f}}_1,...,{\bf{f}}_K]$ with ${\bf f}_{\sf c} \in \mathbb{C}^N$ and ${\bf f}_i \in \mathbb{C}^N$, $\forall i\in\CMcal{K}$ are the $N \times 1$ precoding vectors for the common and private streams, respectively.
The transmit power constraint is defined as $\left\|{\bf{f}}_{\sf c} \right\|^2 + \sum_{i=1}^{K} \left\| {\bf{f}}_i \right\|^2 \le 1$.
The received signal of user $k$ is 
\begin{align} \label{eq:signal_received}
    y_{k} = {\bf{h}}_{k}^{\sf H} {\bf{f}}_{\sf c} s_{\sf c} +{\bf{h}}_{k}^{\sf H} {\bf{f}}_{k} s_{k}  +  \sum_{i = 1, i \neq k}^{K} {\bf{h}}_{k}^{\sf H} {\bf{f}}_{i} s_{i} + z_k,
\end{align}
where $z_k \sim \mathcal{CN}(0,\sigma^2)$ is additive white Gaussian noise (AWGN), and
the channel vector connecting the AP and  user $k$ is represented as ${\bf{h}}_k \in \mathbb{C}^{N}$ where $k \in \CMcal{K}$.
Similarly, the received signal of eavesdropper ${\sf{e}}$ is
\begin{align}
    \label{eq:signal_received_eve}
    y_{\sf{e}} = {\bf{g}}_{\sf{e}}^{\sf H} {\bf{f}}_{\sf c} s_{\sf c} + \sum_{i = 1}^{K} {\bf{g}}_{\sf{e}}^{\sf H} {\bf{f}}_i s_i + z_{\sf{e}},
\end{align}
where $z_{\sf{e}} \sim \mathcal{CN}(0,\sigma_{\sf{e}}^2)$ is AWGN at eavesdropper ${\sf{e}}$, and a channel vector between the AP and the eavesdropper ${\sf{e}}$ is denoted as ${\bf{g}}_{\sf{e}} \in \mathbb{C}^{N}$ where ${\sf{e}} \in \CMcal{E}$.

\section{Problem Formulation}
Initially, each user undertakes the decoding of the common stream $s_c$, treating all remaining private streams $s_k$, $k \in \CMcal{K}$ as interference.
Subsequently, upon the successful decoding of the common stream, employing SIC, users eliminate the common stream from the received signal.
This enables them to decode the private streams with reduced interference levels.

To evaluate the performance metrics associated with RSMA, we formulate the SE  of the common stream $s_{\sf c}$ using  the received signal $y_k$ in \eqref{eq:signal_received}.
All legitimate users should completely decode the common stream part without any error for perfect SIC.
Consequently, the SE of the common stream is formulated as follows:
\begin{align}
    \nonumber
    R_{\sf c}({\bf{F}}) &= \min_{k \in \CMcal{K}}\left\{\log_2 \left(1 + \frac{|{\bf{h}}_{k}^{\sf H} {\bf{f}}_{\sf c}|^2}{\sum_{i \in \CMcal{K}} |{\bf{h}}_{k}^{\sf H} {\bf{f}}_i|^2 + \frac{\sigma^2}{P}} \right)\right\}
    \\ \label{eq:SE_common}
    &= \min_{k \in \CMcal{K}}\left\{ R_{{\sf c},k}({\bf{F}})\right\}.
\end{align}
After perfectly decoding and cancelling the common stream at each user using SIC, the SE of the private stream $s_k$ is 
\begin{align}
    \label{eq:SE_user}
    R_{k}({\bf{F}} ) &= \log_2 \left(1 + \frac{|{\bf{h}}_{k}^{\sf H} {\bf{f}}_{k}|^2}{\sum_{i \in \CMcal{K}\setminus\{k\} } |{\bf{h}}_{k}^{\sf H} {\bf{f}}_i|^2 + \frac{\sigma^2}{P}}\right).
\end{align}
We note that the interference of the common stream term does not exist in the interference term of the formulated SE of the private stream in \eqref{eq:SE_user} due to SIC.

Next, we proceed with computing the leakage SE denoted as $R^{(s)}({\bf{F}})$ pertaining to the private stream of secret users.
The leakage stems not only from eavesdroppers but also from other users.
To expound on several scenarios, with some abuse of notations, we introduce the concept of the leakage set $\CMcal{L}(s) = \CMcal{E} \cup \CMcal{K}\setminus\!\{s\}$.
Then, 
we formulate the maximum achievable leakage SE for secret user $s \in \CMcal{S}$ which is defined as the maximum SE of the symbol of user $s$ among the eavesdroppers and the other legitimate users:
\begin{align} \label{eq:SE_leakage}
    R^{(s)}({\bf{F}}) &= \max_{l \in \CMcal{L}(s)}\{R_l^{(s)}({\bf{F}})\} = \max_{{\sf{e}} \in \CMcal{E}, u \in \CMcal{K}\setminus\{s\}} \left\{R_{\sf{e}}^{(s)}({\bf{F}}),R_u^{(s)}({\bf{F}})\right\},
\end{align}
where $R_{\sf{e}}^{(s)}({\bf{F}})$ is the wiretap SE for user $s$  from eavesdropper ${\sf e}$ and $R_u^{(s)}({\bf{F}})$ is the wiretap SE for user $s$  from the other legitimate user $u$. 
The wiretap SE from eavesdropper $\sf e$ is
\begin{align}
    R_{\sf{e}}^{(s)}({\bf{F}}) = \log_2 \left(1 + \frac{|{\bf{g}}_{\sf{e}}^{\sf H} {\bf{f}}_{s}|^2}{|{\bf{g}}_{\sf{e}}^{\sf H} {\bf{f}}_{\sf{c}}|^2 + \sum_{i \in \CMcal{K}\setminus\{s\}} |{\bf{g}}_{\sf{e}}^{\sf H} {\bf{f}}_i|^2 + \frac{\sigma_{\sf{e}}^2}{P}}\right).
\end{align}
Subsequently,
we present the wiretap SE from legitimate user $u$ as
\begin{align}
    R_u^{(s)}({\bf{F}}) = \log_2 \left(1 + \frac{|{\bf{h}}_{u}^{\sf H} {\bf{f}}_{s}|^2}{\sum_{i \in \CMcal{K}\setminus\{u,s\}} |{\bf{h}}_{u}^{\sf H} {\bf{f}}_i|^2 + \frac{\sigma^2}{P}}\right).
\end{align}

Base on \eqref{eq:SE_common}, \eqref{eq:SE_user}, and \eqref{eq:SE_leakage}, we formulate the sum SE maximization problem as
\begin{align}
    \label{eq:op}
    \mathop{{\text{maximize}}}_{{\bf{f}}_{\sf c}, {\bf{f}}_1, \cdots,{\bf{f}}_K}& \;\; R_{\sf c}({\bf{F}}) \!+\! \sum_{s\in {\CMcal{S}}} \left[R_{s}({\bf{F}}) - R^{(s)}({\bf{F}})\right]^+ \!+\! \sum_{m \in {\CMcal{M}}} R_m({\bf{F}}) \\
    \label{eq:op_txpower}
    {\text{subject to}} & \;\; \left\| {\bf{f}}_{\sf c} \right\|^2 + \sum_{k = 1}^{K} \left\| {\bf{f}}_k \right\|^2  \le 1,
\end{align}
where \eqref{eq:op_txpower} is the transmit power constraint.
It is observed that $R_k({\bf{F}})$ in \eqref{eq:SE_user} is segregated into distinct components: $R_s({\bf{F}})$ for secret users and $R_m({\bf{F}})$ for normal users, where $s \in \CMcal{S}$ and $m \in \CMcal{M}$.
We note that the secrecy SE in \eqref{eq:op} involves the leakage SE due to the legitimate users themselves as well as the eavesdroppers, which requires more careful optimization of the precoder.  
In addition, the problem also includes a non-smooth function due to the rate of the comon stream of  RSMA and is inherently non-convex with respect to the precoder.
In the subsequent sections, addressing the challenges, we solve the problem in \eqref{eq:op} with the assumptions of the perfect CSIT and  limited CSIT where the AP has the imperfect CSIT of legitimate users and partial CSIT of eavesdroppers.

\section{Selective Secure Rate-Splitting Precoding with Perfect CSIT}

In this section, we develop selective secure rate-splitting precoding method with perfect CSIT by resolving the fore-mentioned challenges in the formulated problem.
We reformulate both the SE of the common stream $R_{{\sf c},k}$ in \eqref{eq:SE_common} and private stream $R_k$ in \eqref{eq:SE_user} to adopt the GPI approach \cite{choi2019joint}.
To this end, let us first define the stacked precoding vector as ${\bar{\bf{f}}} = \left[{\bf{f}}_{\sf {c}}^{\sf{T}},{\bf{f}}_1^{\sf{T}},...,{\bf{f}}_{K}^{\sf{T}} \right]^{\sf{T}} \in \bbC^{{N(K+1)}\times 1}$.
Then we assume the maximum transmit power, i.e., $\|\bar{\bf f}\|^2 =1$ as increasing the transmit power in general increases the secrecy SE \cite{bloch2008wireless,oh2023joint}.
Under this transmit power assumption, we reformulate the derived SEs in \eqref{eq:SE_common} and  \eqref{eq:SE_user} as
\begin{align} \label{eq:Rayleigh common and private}
    R_{{\sf c},k}({\bf{\bar{f}}}) = \log_2 \left( \frac{\bar{\bf{f}}^{\sf{H}} {\bf{A}}_{{\sf{c}},k}\bar{\bf{f}}}{\bar{\bf{f}}^{\sf{H}} {\bf{B}}_{{\sf{c}},k}\bar{\bf{f}}} \right),
    R_k({\bf{\bar{f}}}) = \log_2 \left( \frac{\bar{\bf{f}}^{\sf{H}} {\bf{A}}_{k} \bar{\bf{f}}}{\bar{\bf{f}}^{\sf{H}} {\bf{B}}_{k}\bar{\bf{f}}} \right),
\end{align}
where
\begin{align}
    &{\bf{A}}_{{\sf{c}},k} = {\bf{I}}_{K+1} \otimes {\bh}_{k} {\bf{h}}_k^{\sf H} + {\bf{I}}_{N(K+1)}\frac{\sigma^2}{P},\\
    &{\bf{B}}_{{\sf{c}},k} = {\bf{A}}_{{\sf{c}},k} - \rm{diag}\left\{{\bf{e}}_{1}\right\} \otimes {\bh}_{k} {\bf{h}}_k^{\sf H},\\
    &{\bf{A}}_{k} = \left({\bf{I}}_{K+1} - \rm{diag}\left\{{\bf{e}}_{1}\right\}\right) \otimes {\bh}_{k} {\bf{h}}_k^{\sf H} + {\bf{I}}_{N(K+1)}\frac{\sigma^2}{P},\\
    &{\bf{B}}_{k} = {\bf{A}}_{k} - {\rm{diag}}\left\{{\bf{e}}_{k+1}\right\} \otimes {\bh}_{k} {\bf{h}}_k^{\sf H}.
\end{align}

Likewise, the wiretap SE for user $s \in \CMcal{S}$ from eavesdropper ${\sf e}$ can be redefined as
\begin{align} 
    \label{eq:Rayleigh leakage}
    R_{\sf e}^{(s)}({\bf{\bar{f}}}) = \log_2 \left( \frac{\bar{\bf{f}}^{\sf{H}} {\bf{A}}_{\sf{e}}^{(s)}\bar{\bf{f}}}{\bar{\bf{f}}^{\sf{H}} {\bf{B}}_{\sf{e}}^{(s)}\bar{\bf{f}}} \right),
\end{align}
where
\begin{align}
    &{\bf{A}}_{\sf{e}}^{(s)} = {\bf{I}}_{K+1} \otimes {\bg}_{\sf{e}} {\bf{g}}_{\sf{e}}^{\sf H} + {\bf{I}}_{N(K+1)}\frac{\sigma_{\sf{e}}^2}{P},\\
    &{\bf{B}}_{\sf{e}}^{(s)} = {\bf{A}}_{\sf{e}}^{(s)} - {\rm{diag}}\left\{{\bf{e}}_{s+1}\right\} \otimes {\bg}_{\sf{e}} {\bf{g}}_{\sf{e}}^{\sf H},
\end{align}
and the wiretap SE for user $s \in \CMcal{S}$ from  legitimate user $u \in \CMcal{K}\setminus\{s\}$ is reformulated as
\begin{align}
    \label{eq:Rayleigh leakage2}
    R_{u}^{(s)}({\bf{\bar{f}}}) = \log_2 \left( \frac{\bar{\bf{f}}^{\sf{H}} {\bf{C}}_u^{(s)}\bar{\bf{f}}}{\bar{\bf{f}}^{\sf{H}} {\bf{D}}_u^{(s)}\bar{\bf{f}}} \right),
\end{align}
where
\begin{align}
     &{\bf{C}}_{u}^{(s)} = \left({\bf{I}}_{K+1} - {\rm{diag}}\left\{{\bf{e}}_{1} + {\bf{e}}_{u+1}\right\}\right) \otimes {\bf{h}}_{u} {\bf{h}}_u^{\sf H} + {\bf{I}}_{N(K+1)}\frac{\sigma^2}{P},\\
    &{\bf{D}}_{u}^{(s)} = {\bf{C}}_u^{(s)} - {\rm{diag}}\left\{{\bf{e}}_{s+1}\right\} \otimes {\bf{h}}_{u} {\bf{h}}_u^{\sf H}.
\end{align}

Now, to address the non-smoothness challenge, we employ a LogSumExp approach \cite{shen2010dual} for smooth approximation of non-smooth functions: 
\begin{align}
    \label{eq:logsumexp_min}
    \min_{i = 1,...,N}\{x_i\}  \approx -\alpha \log\left(\sum_{i = 1}^{N} \exp\left( \frac{x_i}{-\alpha}  \right)\right),\\
    \label{eq:logsumexp_max}
    \max_{i = 1,...,N}\{x_i\}  \approx \alpha \log\left(\sum_{i = 1}^{N} \exp\left( \frac{x_i}{\alpha}  \right)\right),
\end{align}
where $\alpha >0 $ and the approximation becomes tight as $\alpha \rightarrow +0$.

Based on \eqref{eq:Rayleigh common and private}, \eqref{eq:Rayleigh leakage}, \eqref{eq:Rayleigh leakage2},  \eqref{eq:logsumexp_min}, and \eqref{eq:logsumexp_max}, the original problem in \eqref{eq:op} is reformulated with  relaxing $[\cdot]^+$ as
\begin{align}
    \label{eq:op_Rayleigh}
    \nonumber
    \mathop{{\text{maximize}}}_{{\bf{f}}_{\sf c}, {\bf{f}}_1, \cdots,{\bf{f}}_K}& \;\; \ln\left(\sum_{k \in \CMcal{K}} \left(\frac{\bar{\bf{f}}^{\sf{H}} {\bf{A}}_{{\sf{c}},k}\bar{\bf{f}}}{\bar{\bf{f}}^{\sf{H}} {\bf{B}}_{{\sf{c}},k}\bar{\bf{f}}}\right)^{-\frac{1}{{\alpha}\ln2}}\right)^{-\alpha} + \sum_{s\in {\CMcal{S}}} \left[\log_2 \left( \frac{\bar{\bf{f}}^{\sf{H}} {\bf{A}}_{s} \bar{\bf{f}}}{\bar{\bf{f}}^{\sf{H}} {\bf{B}}_{s}\bar{\bf{f}}} \right) \right. \\
    \nonumber
    & \left. - \ln\left(\sum_{{\sf{e}} \in \CMcal{E}} \left( \frac{\bar{\bf{f}}^{\sf{H}} {\bf{A}}_{\sf{e}}^{(s)}\bar{\bf{f}}}{\bar{\bf{f}}^{\sf{H}} {\bf{B}}_{\sf{e}}^{(s)}\bar{\bf{f}}} \right)^{\frac{1}{\alpha\ln2}}\! \!+\! \!\sum_{u \in \CMcal{K}\setminus\{s\}} \left( \frac{\bar{\bf{f}}^{\sf{H}} {\bf{C}}_u^{(s)}\bar{\bf{f}}}{\bar{\bf{f}}^{\sf{H}} {\bf{D}}_u^{(s)}\bar{\bf{f}}} \right)^{\frac{1}{\alpha\ln2}}\right)^\alpha\right] \\
    &+ \sum_{m\in {\CMcal{M}}} \log_2 \left( \frac{\bar{\bf{f}}^{\sf{H}} {\bf{A}}_{m} \bar{\bf{f}}}{\bar{\bf{f}}^{\sf{H}} {\bf{B}}_{m}\bar{\bf{f}}} \right) \\
    \label{eq:op_Rayleigh_constraint}
    {\text{subject to}} & \;\;  \left\| {\bf{f}} \right\|^2 = 1.
\end{align}

We note that the reformulated optimization problem in \eqref{eq:op_Rayleigh} remains non-convex with respect to the stacked precoding vector, making the search for a global optimal solution still NP-hard. 
Instead, we present an approach to obtain a superior local optimal solution for the reformulated optimization problem.
To delve into this investigation, we first note that the reformulated problem in \eqref{eq:op_Rayleigh} is scale-invariant for the precoder $\bar {\bf f}$. 
Accordingly, we derive the first-order optimality condition of \eqref{eq:op_Rayleigh} by ignoring the power constraint to make the problem more tractable, as outlined in the following lemma:
\begin{lemma}
    \label{lem:KKT}
    With perfect CSIT, the first-order optimality condition of problem \eqref{eq:op_Rayleigh} without  the power constraint is satisfied if the following condition holds:
    \begin{align} \label{eq:KKT_condition}
        {\bf{B}}_{\sf{KKT}}^{-1}(\bar{{\bf{f}}}){\bf{A}}_{\sf{KKT}}(\bar{{\bf{f}}})\bar{{\bf{f}}} = {\lambda}(\bar{{\bf{f}}})\bar{{\bf{f}}},
    \end{align}
    where
    \begin{align} \label{eq:lambda_KKT}
        &\lambda\left(\bar{\bf{f}}\right) = \left[\sum_{k \in \CMcal{K}} w_{{\sf{c}},k}(\bar{\bf{f}})\right]^{-\alpha} \left[\prod_{s \in \CMcal{S}}\left( \frac{\bar{\bf{f}}^{\sf{H}} {\bf{A}}_{s} \bar{\bf{f}}}{\bar{\bf{f}}^{\sf{H}} {\bf{B}}_{s}\bar{\bf{f}}} \right)^{\frac{1}{\ln2}} \right. \\
        \nonumber
        & \left. - \prod_{s \in \CMcal{S}}\left(\sum_{{\sf{e}} \in \CMcal{E}} w_{{\sf{e}}}^{(s)}(\bar{\bf{f}}) + \sum_{u \in \CMcal{K}\setminus\{s\}} w_{u}^{(s)}(\bar{\bf{f}})\right)^\alpha \right] \prod_{m \in \CMcal{M}} \left(\frac{\bar{\bf{f}}^{\sf{H}} {\bf{A}}_{m} \bar{\bf{f}}}{\bar{\bf{f}}^{\sf{H}} {\bf{B}}_{m}\bar{\bf{f}}}\right)^{\frac{1}{\ln2}},
    \end{align}

    \begin{align} \label{eq:A_KKT_condition}
        \nonumber
        &{\bf{A}}_{\sf{KKT}}(\bar{{\bf{f}}}) = \lambda_{\sf{num}}(\bar{\bf{f}})   \left[ \sum_{k \in \CMcal{K}}\left(\frac{w_{{\sf{c}},k}(\bar{\bf{f}})}{\sum_{l \in \CMcal{K}} w_{{\sf{c}},l}(\bar{\bf{f}})}\frac{{{\bf{A}}_{{\sf{c}},k}}}{{\bar{{\bf{f}}}}^{\sf{H}}{\bf{A}}_{{\sf{c}},k}{\bar{{\bf{f}}}}}\right) + \sum_{s\in {\CMcal{S}}} \frac{{\bf{A}}_s}{{\bar{{\bf{f}}}}^{\sf{H}}{\bf{A}}_s{\bar{{\bf{f}}}}} \right. \\
        \nonumber
        & \left. + \sum_{s \in \CMcal{S}} \left(\!\!\frac{\sum_{{\sf{e}} \in \CMcal{E}}\!\left(\!w_{\sf{e}}^{(s)}(\bar{\bf{f}})\!\left(\!\frac{{\bf{B}}_{\sf{e}}^{(s)}}{{\bar{{\bf{f}}}}^{\sf{H}}{\bf{B}}_{\sf{e}}^{(s)}{\bar{{\bf{f}}}}}\right)\right) + \sum_{u \in \CMcal{K}\setminus\{s\}}\left(w_{u}^{(s)}(\bar{\bf{f}})\!\left(\!\frac{{\bf{D}}_u^{(s)}}{{\bar{{\bf{f}}}}^{\sf{H}}{\bf{D}}_u^{(s)}{\bar{{\bf{f}}}}}\right)\right)}{\sum_{i \in \CMcal{E}}w_{i}^{(s)}(\bar{\bf{f}}) + \sum_{i \in \CMcal{K}\setminus\{s\}}w_{i}^{(s)}(\bar{\bf{f}})}\!\!\right) \right. \\
        & \left. + \sum_{m\in {\CMcal{M}}} \frac{{\bf{A}}_m}{{\bar{{\bf{f}}}}^{\sf{H}}{\bf{A}}_m{\bar{{\bf{f}}}}} \right]
    \end{align}

    \begin{align} \label{eq:B_KKT_condition}
        \nonumber
        &{\bf{B}}_{\sf{KKT}}(\bar{{\bf{f}}}) = \lambda_{\sf{den}}(\bar{\bf{f}})  \left[ \sum_{k \in \CMcal{K}}\left(\frac{w_{{\sf{c}},k}(\bar{\bf{f}})}{\sum_{l \in \CMcal{K}}w_{{\sf{c}},l}(\bar{\bf{f}})}\frac{{{\bf{B}}_{{\sf{c}},k}}}{{\bar{{\bf{f}}}}^{\sf{H}}{\bf{B}}_{{\sf{c}},k}{\bar{{\bf{f}}}}}\right) + \sum_{s\in {\CMcal{S}}} \frac{{\bf{B}}_s}{{\bar{{\bf{f}}}}^{\sf{H}}{\bf{B}}_s{\bar{{\bf{f}}}}} \right. \\
        \nonumber
        & \left. + \sum_{s\in \CMcal{S}} \left(\!\!\frac{\sum_{{\sf{e}} \in \CMcal{E}}\!\left(\!w_{{\sf{e}}}^{(s)}(\bar{\bf{f}})\!\left(\!\frac{{\bf{A}}_{\sf{e}}^{(s)}}{{\bar{{\bf{f}}}}^{\sf{H}}{\bf{A}}_{\sf{e}}^{(s)}{\bar{{\bf{f}}}}}\right)\right) + \sum_{u \in \CMcal{K}\setminus\{s\}}\left(w_{u}^{(s)}(\bar{\bf{f}})\!\left(\!\frac{{\bf{C}}_u^{(s)}}{{\bar{{\bf{f}}}}^{\sf{H}}{\bf{C}}_u^{(s)}{\bar{{\bf{f}}}}}\right)\right)}{\sum_{i \in \CMcal{E}}w_{i}^{(s)}(\bar{\bf{f}}) + \sum_{i \in \CMcal{K}\setminus\{s\}}w_{i}^{(s)}(\bar{\bf{f}})}\!\!\right) \right. \\
        & \left. + \sum_{m\in {\CMcal{M}}} \frac{{\bf{B}}_m}{{\bar{{\bf{f}}}}^{\sf{H}}{\bf{B}}_m{\bar{{\bf{f}}}}}\right]
    \end{align}
    with $w_{{\sf{c}},k}(\bar{\bf{f}}) = \left(\frac{\bar{\bf{f}}^{\sf{H}} {\bf{A}}_{{\sf{c}},k}\bar{\bf{f}}}{\bar{\bf{f}}^{\sf{H}} {\bf{B}}_{{\sf{c}},k}\bar{\bf{f}}}\right)^{-\frac{1}{\alpha\ln{2}}}$,
    $w_{{\sf{e}}}^{(s)}(\bar{\bf{f}}) = \left(\frac{\bar{\bf{f}}^{\sf{H}} {\bf{A}}_{\sf{e}}^{(s)}\bar{\bf{f}}}{\bar{\bf{f}}^{\sf{H}} {\bf{B}}_{\sf{e}}^{(s)}\bar{\bf{f}}}\right)^{\frac{1}{\alpha\ln{2}}}$, and
    $w_{u}^{(s)}(\bar{\bf{f}}) = \left(\frac{\bar{\bf{f}}^{\sf{H}} {\bf{C}}_u^{(s)}\bar{\bf{f}}}{\bar{\bf{f}}^{\sf{H}} {\bf{D}}_u^{(s)}\bar{\bf{f}}}\right)^{\frac{1}{\alpha\ln{2}}}$.
    Here, $\lambda_{\sf num}(\bar \bff)$ and $\lambda_{\sf den}(\bar \bff)$ can be any functions of $\bar \bff$ which satisfies $\lambda(\bar \bff) = {\lambda_{\sf{num}}\left(\bar{\bf{f}}\right)}/{\lambda_{\sf{den}}\left(\bar{\bf{f}}\right)}$.
\end{lemma}

\begin{proof}
    See Appendix \ref{app:lemma1}.
\end{proof}

We treat \eqref{eq:KKT_condition} as an eigenvector-dependent nonlinear eigenvalue problem (NEPv) \cite{cai2018eigenvector}.
We remark that the  Lagrangian function  $L(\bar {\bf f})$ is  also reformulated as $L(\bar \bff) = \ln \lambda(\bar\bff)$ where $\lambda(\bar \bff)$ is defined in \eqref{eq:lambda_KKT}.
From this relationship, this problem is essentially about finding the principal eigenvector $\bar {\bf{f}}^\star$ that maximizes $\lambda(\bar {\bf{f}})$.
In order to tackle this, we adopt the GPI method \cite{choi2019joint}, leading to the development of the proposed algorithm.
The process, referred to as selective secure SE RSMA precoding based on GPI (SSSE-GPI-RSMA), is described in Algorithm~\ref{algorithm:1}.
It iteratively updates $\bar {\bf f}$ using simple projection and normalization until convergence.
We note that the feasibility can be guaranteed via normalization.

\begin{algorithm} [t]
    \caption{SSSE-GPI-RSMA} 
    \label{algorithm:1}
    {\bf{initialize}}: $\bar {\bf{f}}_{0}$\\
    Set the iteration count $t = 0$.\\
    \While {$\left\|\bar {\bf{f}}_{t+1} - \bar {\bf{f}}_{t} \right\| > \epsilon$ $\it{\&}$ $t \leq t_{\rm max}$}{
    Matrix construction $ {\bf{A}}_{\sf KKT} (\bar {\bf{f}}_{t})$ in \eqref{eq:A_KKT_condition} \\
    Matrix construction  $ {\bf{B}}_{\sf KKT} (\bar {\bf{f}}_{t})$ in \eqref{eq:B_KKT_condition} \\
    Projection $\bar {\bf{f}}_{t+1} = $ ${{\bf{B}}^{-1}_{\sf KKT} (\bar {\bf{f}}_{t}) {\bf{A}}_{\sf KKT} (\bar {\bf{f}}_{t}) \bar {\bf{f}}_{t}}$. \\
    Normalization $\bar {\bf{f}}_{t+1} = \bar {\bf{f}}_{t+1}/\|\bar {\bf{f}}_{t+1}\|.$
    \\
     $t \leftarrow t+1$.}
    \Return{\ }{$\bar{\bf f}_t$}.
\end{algorithm}

\section{Extension to Limited CSIT: Robust and Selective Secure Rate-Splitting  Precoding}

In this section, we introduce the consideration of imperfect and partial CSIT to broaden the applicability of our proposed algorithm towards more realistic scenarios.
Specifically, the AP is restricted to having limited access solely to the estimated channel for each user $k$, defined as
\begin{align} \label{eq:estimated_ch}
        \hat{\bf{h}}_k = {\bf{h}}_k - {\pmb{\phi}}_k,
\end{align}
where $\hat{\bf{h}}_k$ is the estimated channel vector of user $k$ and ${\pmb{\phi}}_k$ is the CSIT estimation error vector. 
Furthermore, the AP also has access to the spatial channel covariance matrix ${\bf{R}}_k = \mathbb{E}[{\bf{h}}_k {\bf{h}}_k^{\sf H}]$ and associated estimation error covariance matrix ${\bf{\Phi}}_k = \mathbb{E}[{\pmb{\phi}}_{k} {\pmb{\phi}}_k^{\sf H}]$.
Regarding the wiretap channels, the knowledge of AP is confined to the wiretap channel covariance matrix ${\bf{R}}_{\sf{e}} = \mathbb{E}[{\bf{g}}_{\sf{e}} {\bf{g}}_{\sf{e}}^{\sf H}]$, not the estimated channel ${\bf g}_{\sf e}$.

\subsection{Conditional Average Rate}
We assume imperfect and partial CSIT so as to consider more practical scenarios.
AP is able to access to the estimated channel of users with a known estimation error covariance and the wiretap channel covariance matrix.
Since only the estimated legitimate channel and other long-term channel statistics are available at the AP, we adopt a conditional averaged secrecy SE approach \cite{joudeh2016sum,salem2022secure} to convert the problem for exploiting the long-term channel statistics.
To this end, we define $\mathbb{E}\left[R_{{\sf{c}},k}({\bf{F}})\right]$, $\mathbb{E}\left[R_k({\bf{F}})\right]$, and $\mathbb{E}\left[{R}^{(s)}({\bf{F}})\right]$ as the ergodic SEs of the common stream, private stream of user $k$, and maximum leakage for the stream $s_s$, respectively, \cite{joudeh2016sum,salem2022secure} and replace $R_{{\sf{c}},k}({\bf{F}})$, $R_k({\bf{F}})$, and $R^{(s)}({\bf{F}})$ in \eqref{eq:SE_common}, \eqref{eq:SE_user}, and \eqref{eq:SE_leakage} with the corresponding ergodic SEs.

Then, the formulated the ergodic SE of the common stream $\bar{R}_{\sf{c}}$ is expressed as \cite{joudeh2016sum}
\begin{align}
    \nonumber
    \label{eq:SE_common_partial}
    &\bar{R}_{\sf c}({\bf{F}}) \\
    \nonumber
    &= \min_{k \in \CMcal{K}} \left\{ \mathbb{E}_{{\bf h}_k} \left[\log_2 \left(1 + \frac{|{\bf{h}}_{k}^{\sf H} {\bf{f}}_{\sf c}|^2}{\sum_{i \in \CMcal{K}} |{\bf{h}}_{k}^{\sf H} {\bf{f}}_i|^2 + \frac{\sigma^2}{P}} \right)   \right] \right\} \\
    \nonumber
    &\stackrel{(a)}= \min_{k \in \CMcal{K}} \left\{ \mathbb{E}_{\hat{\bf h}_k}\left[\mathbb{E}_{{\bf h}_k|{\hat{\bf h}}_k} \left[\log_2 \left. \left(1 + \frac{|{\bf{h}}_{k}^{\sf H} {\bf{f}}_{\sf c}|^2}{\sum_{i \in \CMcal{K}} |{\bf{h}}_{k}^{\sf H} {\bf{f}}_i|^2 + \frac{\sigma^2}{P}} \right) \right| \hat{{\bf{h}}}_k \right] \right] \right\} \\
    &= \min_{k \in \CMcal{K}}\{\bar{R}_{{\sf c},k}({\bf{F}})\},
\end{align}
where $(a)$ comes from the law of total expectation.
We also derive the ergodic SE  of private stream for user $k$ as \cite{joudeh2016sum}
\begin{align}
    \label{eq:SE_private_partial}
    \bar{R}_k({\bf{F}}) &= \mathbb{E}_{{\bf h}_k} \left[\log_2 \left(1 + \frac{|{\bf{h}}_{k}^{\sf H} {\bf{f}}_k|^2}{\sum_{i \in \CMcal{K}\setminus\{k\}} |{\bf{h}}_{k}^{\sf H} {\bf{f}}_i|^2 + \frac{\sigma^2}{P}} \right) \right] \\
    \nonumber
    &= \mathbb{E}_{\hat{\bf h}_k} \left[\mathbb{E}_{{\bf h}_k|{\hat{\bf h}}_k} \left[\log_2 \left. \left(1 + \frac{|{\bf{h}}_{k}^{\sf H} {\bf{f}}_k|^2}{\sum_{i \in \CMcal{K}\setminus\{k\}} |{\bf{h}}_{k}^{\sf H} {\bf{f}}_i|^2 + \frac{\sigma^2}{P}} \right) \right| \hat{{\bf{h}}}_k \right] \right].
\end{align}

Similarly, the ergodic maximum  leakage SE $\bar{R}^{(s)}$ \cite{salem2022secure} is  
\begin{align} \label{eq:SE_leakage_partial}
    \nonumber
    \bar{R}^{(s)}({\bf{F}}) &=\mathbb{E}_{{{\bf h}_u},{{\bf g}_{\sf{e}}}} \left[ R^{(s)}({\bf{F}})\right] 
    \\
    &= \mathbb{E}_{{{\bf h}_u},{{\bf g}_{\sf{e}}}} \left[\max_{l \in \CMcal{L}(s)}R_l^{(s)}({\bf{F}})\right] 
    \\
    &= \mathbb{E}_{{{\bf h}_u},{{\bf g}_{\sf{e}}}} \left[\max_{{\sf{e}} \in \CMcal{E}, u \in \CMcal{K} \setminus\{s\}} \left\{R_{\sf{e}}^{(s)}({\bf{F}}),R_u^{(s)}({\bf{F}})\right\} \right].
\end{align}

Using the derived ergodic SE in \eqref{eq:SE_common_partial}, \eqref{eq:SE_private_partial}, and \eqref{eq:SE_leakage_partial}, we now rewrite the optimization problem in \eqref{eq:op} considering the imperfect CSIT in the following manner:
\begin{align}
    \nonumber
    \mathop{{\text{maximize}}}_{{\bf{f}}_{\sf c}, {\bf{f}}_1, \cdots,{\bf{f}}_K}& \;\; \bar{R}_{\sf c}({\bf{F}}) + \sum_{s\in {\CMcal{S}}} \left[\bar{R}_{s}({\bf{F}}) - \bar{R}^{(s)}({\bf{F}})\right]^+ + \sum_{m \in {\CMcal{M}}} \bar{R}_m({\bf{F}}) \\
    {\text{subject to}} & \;\; \left\| {\bf{f}}_{\sf c} \right\|^2 + \sum_{k = 1}^{K} \left\| {\bf{f}}_k \right\|^2  \le 1 .
    \label{eq:op_partial}
\end{align}


\subsection{Limited CSIT Handling Appraoch}
As the optimization problem in \eqref{eq:op_partial} cannot be directly solved with the limited information available, we instead establish the lower bounds for each SE.
The lower bound of the ergodic SE of the common stream  $\bar{R}_{{\sf{c}},k}({\bf F})$ in \eqref{eq:SE_common_partial} is derived as follows:
\begin{align} \label{eq:SE_common_partial_lb}
    \nonumber
    &\bar{R}_{{\sf c},k}({\bf F}) \\
    \nonumber
    &\stackrel{(b)}\geq \!\! \mathbb{E}_{\hat{\bf h}_k}\! \left[\log_2 \left(1\! +\! \frac{|\hat{{\bf{h}}}_{k}^{\sf H} {\bf{f}}_{\sf c}|^2}{\sum_{i \in \CMcal{K}} |\hat{{\bf{h}}}_{k}^{\sf H} {\bf{f}}_i|^2 \!+ \!{\bf{f}}^{\sf{H}}_{\sf{c}}{\bf{\Phi}}_k{\bf{f}}_{\sf{c}} \!+ \!\sum_{i \in \CMcal{K}} {\bf{f}}^{\sf{H}}_i{\bf{\Phi}}_k{\bf{f}}_i \!+ \!\frac{\sigma^2}{P}} \right) \right] \\
    &= \mathbb{E}_{\hat{\bf h}_k} \left[R^{\sf{lb}}_{{\sf c},k}({\bf F};\hat{\bf h}_k, {\bf \Phi}_k)\right],
\end{align}
where $(b)$ is obtained by assuming that all error terms represent independent Gaussian noise, according to the principles outlined in the generalized mutual information framework \cite{yoo2006capacity}.
Subsequently, the derivation uses Jensen's inequality to reach its conclusion.
Then, based on \eqref{eq:SE_common_partial} and \eqref{eq:SE_common_partial_lb}, we further have
\begin{align}
    \label{eq:SE_common_partial_lb_v2}
    \nonumber
    \bar{R}_{\sf{c}}(\bf F) &= \min_{k \in \CMcal{K}} \left\{\bar{R}_{{\sf c},k}(\bf F) \right\} 
    \\
    &\geq \min_{k \in \CMcal{K}} \left\{\mathbb{E}_{\hat{\bf h}_k} \left[R^{\sf{lb}}_{{\sf c},k}({\bf F};\hat{\bf h}_k, {\bf \Phi}_k)\right] \right\}
    \\
    & \geq \mathbb{E}_{\hat{\bf h}_k} \left[\min_{k \in \CMcal{K}} \left\{R^{\sf{lb}}_{{\sf c},k}({\bf F};\hat{\bf h}_k, {\bf \Phi}_k) \right\} \right].
\end{align}

Likewise, we can derive the lower bound of the ergodic SE of the private stream $\bar{R}_k({\bf F})$ in \eqref{eq:SE_private_partial} as follows:
\begin{align} \label{eq:SE_private_partial_lb}
    \nonumber
    \bar{R}_k({\bf F}) 
    &\!\stackrel{(c)}\geq \!\mathbb{E}_{\hat{\bf h}_k}  \!\left[\log_2 \left(1 \! +  \!\frac{|\hat{{\bf{h}}}_{k}^{\sf H} {\bf{f}}_k|^2}{\sum_{i \in \CMcal{K}\setminus\{k\}} |\hat{{\bf{h}}}_{k}^{\sf H} {\bf{f}}_i|^2  \!+ \! \sum_{i \in \CMcal{K}} {\bf{f}}^{\sf{H}}_i{\bf{\Phi}}_k{\bf{f}}_i \! +  \!\frac{\sigma^2}{P}} \right) \right]\\
    &= \mathbb{E}_{\hat{\bf h}_k} \left[ R^{\sf{lb}}_k({\bf F};\hat{\bf h}_k, {\bf \Phi}_k)\right],
\end{align}
where $(c)$ stems from Jensen's inequality.

Now, to derive the ergodic maiximum leakage SE  
$\bar{R}^{(s)}({\bf{F}})$ \eqref{eq:SE_leakage_partial} in a tractable form,
we introduce the following lemma:
\begin{lemma}
    \label{sum_approx}
    Let random variables ($X_i, Y_i$) be independent with ($X_j, Y_j$)  for $i \neq j$.
    Then the following approximation holds with large $N$:
    \begin{align}
        \label{lemma_approx}
        \mathbb{E} \left[\log_2 \left(1 + \sum_{i=1}^N\frac{X_i}{Y_i}\right)\right] \approx \log_2 \left(1 + \sum_{i=1}^N \frac{\mathbb{E} \left[X_i\right]}{\mathbb{E} \left[Y_i\right]}\right).
    \end{align}
\end{lemma}

\begin{proof}
    See Appendix~\ref{app:lemma2}.
\end{proof}

Then, to properly derive a useful bound for $\bar{R}^{(s)}({\bf{F}})$ \eqref{eq:SE_leakage_partial},  we consider the worst-case scenario for secure communication in which all the other users and eavesdroppers engage in full collusion. 
In this case, the wiretap SE is upper bounded as \cite{yang2016physical, choi2021sum}
\begin{align}
    \nonumber
    \bar{R}^{(s)}({\bf{F}}) 
    &\leq \mathbb{E} \left[\log_2 \left(1 + \sum_{{\sf{e}} \in \CMcal{E}} \frac{|{\bf{g}}_{\sf{e}}^{\sf H} {\bf{f}}_{s}|^2}{|{\bf{g}}_{\sf{e}}^{\sf H} {\bf{f}}_{\sf{c}}|^2 + \sum_{i \in \CMcal{K}\setminus\{s\}} |{\bf{g}}_{\sf{e}}^{\sf H} {\bf{f}}_i|^2 + \frac{\sigma_{\sf{e}}^2}{P}} \right. \right. 
    \\
    &\quad + \left. \left. \sum_{u \in \CMcal{K}\setminus\!\{s\}}\frac{|{\bf{h}}_{u}^{\sf H} {\bf{f}}_{s}|^2}{\sum_{i \in \CMcal{K}\setminus\{u,s\}} |{\bf{h}}_{u}^{\sf H} {\bf{f}}_i|^2 + \frac{\sigma^2}{P}}\right) \right]
    \\\nonumber
     &\stackrel{(d)}\approx  \log_2 \left(1 + \sum_{{\sf{e}} \in \CMcal{E}} \frac{{\bf{f}}^{\sf{H}}_{s}{\bf{R}}_{\sf{e}}{\bf{f}}_{s}}{{\bf{f}}^{\sf{H}}_{\sf{c}}{\bf{R}}_{\sf{e}}{\bf{f}}_{\sf{c}} + \sum_{i \in \CMcal{K}\setminus\{s\}} {\bf{f}}^{\sf{H}}_{i}{\bf{R}}_{\sf{e}}{\bf{f}}_{i} + \frac{\sigma_{\sf{e}}^2}{P}} \right. 
    \\
    &\quad + \left. \sum_{u \in \CMcal{K}\setminus\!\{s\}}\frac{{\bf{f}}^{\sf{H}}_{s}{\bf{R}}_u{\bf{f}}_{s}}{\sum_{i \in \CMcal{K}\setminus\{u,s\}} {\bf{f}}^{\sf{H}}_{i}{\bf{R}}_u{\bf{f}}_{i} + \frac{\sigma^2}{P}}\right) 
    \\\label{eq:Rs_ub}
    &= R^{(s){\sf{ub}}}({\bf F};{\bf R}_{\sf{e}},{\bf R}_u),
\end{align}
where $(d)$ is derived from  Lemma \ref{sum_approx}.
We note that \eqref{eq:Rs_ub} includes the channel covariance matrices of users and eavesdroppers which are available at the AP and thus, leveraging \eqref{eq:Rs_ub} allows us to adequately incorporate the leakage SE in the optimization with the limited CSIT.

Finally, using \eqref{eq:SE_common_partial_lb_v2}, \eqref{eq:SE_private_partial_lb}, and \eqref{eq:Rs_ub}  we can derive the lower bound of the ergodic secrecy sum SE in \eqref{eq:op_partial} as follows:
\begin{align}  
    \nonumber
     &\bar{R}_{\sf{c}}({\bf F}) + \sum_{s\in {\CMcal{S}}} \left[\bar{R}_{s}({\bf F}) - \bar{R}^{(s)}({\bf{F}}) \right]^+ + \sum_{m \in {\CMcal{M}}} \bar{R}_m({\bf F}) \\
    \nonumber
    & \geq \bbE_{\hat{\bf h}_k}\left[\min_{k \in \CMcal{K}} \left\{ R^{\sf{lb}}_{{\sf{c}},k}({\bf F};\hat{\bf h}_k, {\bf \Phi}_k)\right\} \right]+ \sum_{s\in {\CMcal{S}}} \left[\bbE_{\hat{\bf h}_k}\left[R^{\sf{lb}}_s({\bf F};\hat{\bf h}_s, {\bf \Phi}_s) \right]\right.  
    \\\label{eq:ergodic_SE_lb}
    & \quad  \left. - R^{(s){\sf{ub}}}({\bf F};{\bf R}_{\sf{e}},{\bf R}_u) \right]^+ \! + \! \bbE_{\hat{\bf h}_k}\left[ \sum_{m \in {\CMcal{M}}} R^{\sf{lb}}_m ({\bf F};\hat{\bf h}_m, {\bf \Phi}_m) \right].
\end{align}
We note that since the estimated channel information $\hat {\bf h}_k$ is available at the AP, we can use \eqref{eq:ergodic_SE_lb} without the expectation over $\hat {\bf h}_k$ as a new objective function. 
Then, relaxing the non-zero operator $[\cdot]^+$, we reformulate the lower bound of the optimization problem as follows:
\begin{align}
    \nonumber
    \mathop{{\text{maximize}}}_{{\bf{f}}_{\sf c}, {\bf{f}}_1, \cdots,{\bf{f}}_K}& \;\; \min_{k \in \CMcal{K}} \left\{ R^{\sf{lb}}_{{\sf{c}},k}({\bf F};\hat{\bf h}_k, {\bf \Phi}_k) \right\} + \sum_{s\in {\CMcal{S}}} \left(R^{\sf{lb}}_s({\bf F};\hat{\bf h}_s, {\bf \Phi}_s) \right.
    \\\label{eq:problem_imperfect_lb}
    &  \left.- R^{(s){\sf{ub}}}({\bf F};{\bf R}_{\sf{e}},{\bf R}_u)  \right)+ \sum_{m \in {\CMcal{M}}} R^{\sf{lb}}_m ({\bf F};\hat{\bf h}_m, {\bf \Phi}_m) 
    \\
    {\text{subject to}} & \;\;  \left\| {\bf{f}}_{\sf c} \right\|^2 + \sum_{k = 1}^{K} \left\| {\bf{f}}_k \right\|^2  \le 1.
    \label{reformulatedop}
\end{align}
In the following subsection, we use the same approach as the perfect CSIT case to identify the superior local optimal solution of the reformulated problem in \eqref{eq:problem_imperfect_lb}.

\subsection{Problem Reformulation}
Assuming $\|\bar {\bf f}\|^2=1$ where $\bar {\bf f} = {\rm vec}({\bf F})$, we first cast $ R^{\sf{lb}}_{{\sf{c}},k}({\bf F})$ and $ R^{\sf{lb}}_k({\bf F})$ as
\begin{align} \label{eq:Rayleigh common and private_partial}
    R^{\sf{lb}}_{{\sf{c}},k}({\bf{\bar{f}}}) = \log_2 \left( \frac{\bar{\bf{f}}^{\sf{H}} {\bf{A}}_{{\sf{c}},k}^{\sf{lb}}\bar{\bf{f}}}{\bar{\bf{f}}^{\sf{H}} {\bf{B}}_{{\sf{c}},k}^{\sf{lb}}\bar{\bf{f}}} \right),
    R^{\sf{lb}}_k({\bf{\bar{f}}}) = \log_2 \left( \frac{\bar{\bf{f}}^{\sf{H}} {\bf{A}}_{k}^{\sf{lb}}\bar{\bf{f}}}{\bar{\bf{f}}^{\sf{H}}{\bf{B}}_{k}^{\sf{lb}}\bar{\bf{f}}} \right),
\end{align}
where
\begin{align}
    &{\bf{A}}_{{\sf{c}},k}^{\sf{lb}} = {\bf{I}}_{K+1} \otimes \left(\hat{{\bh}}_{k} \hat{{\bf{h}}}_k^{\sf H} + {\bf{\Phi}}_k \right) + {\bf{I}}_{N(K+1)}\frac{\sigma^2}{P},\\
    &{\bf{B}}_{{\sf{c}},k}^{\sf{lb}} = {\bf{A}}_{{\sf{c}},k}^{\sf{lb}} - {\rm{diag}}\left\{{\bf{e}}_{1}\right\} \otimes \hat{{\bh}}_{k} \hat{{\bf{h}}}_k^{\sf H},\\
    &{\bf{A}}_{k}^{\sf{lb}} = \left({\bf{I}}_{K+1} - {\rm{diag}}\left\{{\bf{e}}_{1}\right\}\right) \otimes \left(\hat{{\bf{h}}}_{k} \hat{{\bf{h}}}_k^{\sf H} + {\bf{\Phi}}_k\right) + {\bf{I}}_{N(K+1)}\frac{\sigma^2}{P},\\
    &{\bf{B}}_{k}^{\sf{lb}} = {\bf{A}}_{k}^{\sf{lb}} - {\rm{diag}}\left\{{\bf{e}}_{k+1}\right\} \otimes \hat{{\bf{h}}}_{k} \hat{{\bf{h}}}_k^{\sf H}.
\end{align}
Similarly, we reformulate the average leakage SE as
\begin{align} \label{eq:Rayleigh leakage_partial}
    R^{(s){\sf{ub}}}({\bf{\bar{f}}}) \!= \!\log_2 \left( 1 \!+ \!\sum_{{\sf{e}} \in \CMcal{E}}\frac{\bar{\bf{f}}^{\sf{H}} {\bf{A}}^{(s){\sf{ub}}}\bar{\bf{f}}}{\bar{\bf{f}}^{\sf{H}} {\bf{B}}^{(s){\sf{ub}}}\bar{\bf{f}}} + \!\!\!\!\sum_{u \in \CMcal{K}\setminus\{s\}} \!\!\frac{\bar{\bf{f}}^{\sf{H}} {\bf{C}}^{(s){\sf{ub}}}\bar{\bf{f}}}{\bar{\bf{f}}^{\sf{H}} {\bf{D}}^{(s){\sf{ub}}}\bar{\bf{f}}}\right),
\end{align}
where
\begin{align}
    &{\bf{A}}^{(s){\sf{ub}}} = {\rm{diag}}\left\{{\bf{e}}_{s+1}\right\} \otimes {\bR}_{\sf{e}}, \\
    &{\bf{B}}^{(s){\sf{ub}}} = {\bf{I}}_{K+1} \otimes {\bR}_{\sf{e}} - {\bf{A}}^{(s){\sf{ub}}} + {\bf{I}}_{N(K+1)}\frac{\sigma_{\sf{e}}^2}{P}, \\
    &{\bf{C}}^{(s){\sf{ub}}} = {\rm{diag}}\left\{{\bf{e}}_{s+1}\right\} \otimes {\bR}_u, \\
    &{\bf{D}}^{(s){\sf{ub}}} = \left({\bf{I}}_{K+1} - {\rm{diag}}\left\{{\bf{e}}_{1} + {\bf{e}}_{u+1}\right\}\right) \otimes {\bR}_u - {\bf{C}}^{(s){\sf{ub}}} + {\bf{I}}_{N(K+1)}\frac{\sigma^2}{P}.
\end{align}

Then, the rest of the steps are the same as the perfect CSIT case: we approximate the maximum operator in \eqref{eq:problem_imperfect_lb} with the LogSumExp approach and further derive the first-order KKT condition of the problem as in Lemma~\ref{lem:KKT_lb}. 
\begin{lemma}
    \label{lem:KKT_lb}
    With limited CSIT, the first-order optimality condition of problem \eqref{reformulatedop} without  the power constraint is satisfied if the following condition holds:
    \begin{align} \label{eq:KKT_condition_lb}
        ({\bf{B}}_{\sf{KKT}}^{\sf lb})^{-1}(\bar{{\bf{f}}}){\bf{A}}^{\sf lb}_{\sf{KKT}}(\bar{{\bf{f}}})\bar{{\bf{f}}} = {\lambda}^{\sf lb}(\bar{{\bf{f}}})\bar{{\bf{f}}},
    \end{align}
    where
    \begin{align} \label{eq:lambda_KKT_2}
        &\lambda^{\sf lb}\left(\bar{\bf{f}}\right) = \left[\sum_{k \in \CMcal{K}} w_{{\sf{c}},k}^{\sf{lb}}(\bar{\bf{f}})\right]^{-\alpha} \left[\prod_{s \in \CMcal{S}}\left( \frac{\bar{\bf{f}}^{\sf{H}} {\bf{A}}_{s}^{\sf{lb}} \bar{\bf{f}}}{\bar{\bf{f}}^{\sf{H}} {\bf{B}}_{s}^{\sf{lb}}\bar{\bf{f}}} \right)^{\frac{1}{\ln2}} - \prod_{s \in \CMcal{S}}\left(1 + \right. \right. \\
        \nonumber
        & \left. \left. \sum_{{\sf{e}} \in \CMcal{E}} w_{{\sf{e}}}^{(s){\sf{ub}}}(\bar{\bf{f}}) + \sum_{u \in \CMcal{K}\setminus\{s\}} w_{u}^{(s){\sf{ub}}}(\bar{\bf{f}})\right)^{\frac{1}{\ln2}} \right] \prod_{m \in \CMcal{M}} \left(\frac{\bar{\bf{f}}^{\sf{H}} {\bf{A}}_{m}^{\sf{lb}} \bar{\bf{f}}}{\bar{\bf{f}}^{\sf{H}} {\bf{B}}_{m}^{\sf{lb}}\bar{\bf{f}}}\right)^{\frac{1}{\ln2}},
    \end{align}
    \begin{align} \label{eq:A_KKT_condition_lb}
        \nonumber
        &{\bf{A}}^{\sf lb}_{\sf{KKT}}(\bar{{\bf{f}}}) = \lambda^{\sf lb}_{\sf{num}}(\bar{\bf{f}}) \times  \left[ \sum_{k \in \CMcal{K}}\left(\frac{w_{{\sf{c}},k}^{\sf{lb}}(\bar{\bf{f}})}{\sum_{l \in \CMcal{K}} w_{{\sf{c}},l}^{\sf{lb}}(\bar{\bf{f}})}\frac{{{\bf{A}}_{{\sf{c}},k}^{\sf{lb}}}}{{\bar{{\bf{f}}}}^{\sf{H}}{\bf{A}}_{{\sf{c}},k}^{\sf{lb}}{\bar{{\bf{f}}}}}\right) \right. \\
        \nonumber
        & \left. + \sum_{s\in {\CMcal{S}}} \frac{{\bf{A}}_s^{\sf{lb}}}{{\bar{{\bf{f}}}}^{\sf{H}}{\bf{A}}_s^{\sf{lb}}{\bar{{\bf{f}}}}} + \sum_{s \in \CMcal{S}} \left\{\left(\!\!\frac{\sum_{{\sf{e}} \in \CMcal{E}}\!\left(\!w_{\sf{e}}^{(s){\sf{lb}}}(\bar{\bf{f}})\!\left(\!\frac{{\bf{B}}_{\sf{e}}^{(s){\sf{ub}}}}{{\bar{{\bf{f}}}}^{\sf{H}}{\bf{B}}_{\sf{e}}^{(s){\sf{ub}}}{\bar{{\bf{f}}}}}\right)\right)}{1 + \sum_{i \in \CMcal{E}}w_i^{(s){\sf{ub}}}(\bar{\bf{f}}) + \sum_{i \in \CMcal{K}\setminus\{s\}}w_{i}^{(s){\sf{ub}}}(\bar{\bf{f}})}\!\!\right) \right. \right. \\
        & \left. \left. + \left(\!\!\frac{\sum_{u \in \CMcal{K}\setminus\{s\}}\left(w_{u}^{(s){\sf{ub}}}(\bar{\bf{f}})\!\left(\!\frac{{\bf{D}}_u^{(s){\sf{ub}}}}{{\bar{{\bf{f}}}}^{\sf{H}}{\bf{D}}_u^{(s){\sf{ub}}}{\bar{{\bf{f}}}}}\right)\right)}{1 + \sum_{i \in \CMcal{E}}w_i^{(s){\sf{ub}}}(\bar{\bf{f}}) + \sum_{i \in \CMcal{K}\setminus\{s\}}w_{i}^{(s){\sf{ub}}}(\bar{\bf{f}})}\!\!\right) \right\} + \sum_{m\in {\CMcal{M}}} \frac{{\bf{A}}_m^{\sf{lb}}}{{\bar{{\bf{f}}}}^{\sf{H}}{\bf{A}}_m^{\sf{lb}}{\bar{{\bf{f}}}}} \right].
    \end{align}
    and
    \begin{align} \label{eq:B_KKT_condition_lb}
        \nonumber
        &{\bf{B}}^{\sf lb}_{\sf{KKT}}(\bar{{\bf{f}}}) = \lambda^{\sf lb}_{\sf{den}}(\bar{\bf{f}}) \times  \left[ \sum_{k \in \CMcal{K}}\left(\frac{w_{{\sf{c}},k}^{\sf{lb}}(\bar{\bf{f}})}{\sum_{l \in \CMcal{K}} w_{{\sf{c}},l}^{\sf{lb}}(\bar{\bf{f}})}\frac{{{\bf{B}}_{{\sf{c}},k}^{\sf{lb}}}}{{\bar{{\bf{f}}}}^{\sf{H}}{\bf{B}}_{{\sf{c}},k}^{\sf{lb}}{\bar{{\bf{f}}}}}\right) \right. \\
        \nonumber
        & \left. + \sum_{s\in {\CMcal{S}}} \frac{{\bf{B}}_s^{\sf{lb}}}{{\bar{{\bf{f}}}}^{\sf{H}}{\bf{B}}_s^{\sf{lb}}{\bar{{\bf{f}}}}} + \sum_{s \in \CMcal{S}} \left\{\left(\!\!\frac{\sum_{{\sf{e}} \in \CMcal{E}}\!\left(\!w_{\sf{e}}^{(s){\sf{lb}}}(\bar{\bf{f}})\!\left(\!\frac{{\bf{A}}_{\sf{e}}^{(s){\sf{ub}}}}{{\bar{{\bf{f}}}}^{\sf{H}}{\bf{A}}_{\sf{e}}^{(s){\sf{ub}}}{\bar{{\bf{f}}}}}\right)\right)}{1 + \sum_{i \in \CMcal{E}}w_i^{(s){\sf{ub}}}(\bar{\bf{f}}) + \sum_{i \in \CMcal{K}\setminus\{s\}}w_{i}^{(s){\sf{ub}}}(\bar{\bf{f}})}\!\!\right) \right. \right. \\
        & \left. \left. + \left(\!\!\frac{\sum_{u \in \CMcal{K}\setminus\{s\}}\left(w_{u}^{(s){\sf{ub}}}(\bar{\bf{f}})\!\left(\!\frac{{\bf{C}}_u^{(s){\sf{ub}}}}{{\bar{{\bf{f}}}}^{\sf{H}}{\bf{C}}_u^{(s){\sf{ub}}}{\bar{{\bf{f}}}}}\right)\right)}{1 + \sum_{i \in \CMcal{E}}w_i^{(s){\sf{ub}}}(\bar{\bf{f}}) + \sum_{i \in \CMcal{K}\setminus\{s\}}w_{i}^{(s){\sf{ub}}}(\bar{\bf{f}})}\!\!\right) \right\} + \sum_{m\in {\CMcal{M}}} \frac{{\bf{B}}_m^{\sf{lb}}}{{\bar{{\bf{f}}}}^{\sf{H}}{\bf{B}}_m^{\sf{lb}}{\bar{{\bf{f}}}}} \right].
    \end{align}
    with $w_{{\sf{c}},k}^{\sf{lb}}(\bar{\bf{f}}) = \left(\frac{\bar{\bf{f}}^{\sf{H}} {\bf{A}}_{{\sf{c}},k}^{\sf{lb}}\bar{\bf{f}}}{\bar{\bf{f}}^{\sf{H}} {\bf{B}}_{{\sf{c}},k}^{\sf{lb}}\bar{\bf{f}}}\right)^{-\frac{1}{\alpha\ln{2}}}$,
    $w_{{\sf{e}}}^{(s){\sf{ub}}}(\bar{\bf{f}}) = \frac{\bar{\bf{f}}^{\sf{H}} {\bf{A}}^{(s){\sf{ub}}} \bar{\bf{f}}}{\bar{\bf{f}}^{\sf{H}} {\bf{B}}^{(s){\sf{ub}}}\bar{\bf{f}}}$, and
    $w_{u}^{(s){\sf{ub}}}(\bar{\bf{f}}) = \frac{\bar{\bf{f}}^{\sf{H}} {\bf{C}}^{(s){\sf{ub}}} \bar{\bf{f}}}{\bar{\bf{f}}^{\sf{H}} {\bf{D}}^{(s){\sf{ub}}}\bar{\bf{f}}}$.
    Here, $\lambda^{\sf lb}_{\sf num}(\bar \bff)$ and $\lambda^{\sf lb}_{\sf den}(\bar \bff)$ can be any functions of $\bar \bff$ which satisfies $\lambda^{\sf lb}(\bar \bff) = {\lambda^{\sf lb}_{\sf{num}}\left(\bar{\bf{f}}\right)}/{\lambda^{\sf lb}_{\sf{den}}\left(\bar{\bf{f}}\right)}$.
\end{lemma}

\begin{proof}
    We omit the detailed proof here as it is similar to the one of Lemma~\ref{lem:KKT}.
\end{proof}
We remark that  Algorithm~\ref{algorithm:1} can be also used to find the principal eigenvector of \eqref{eq:KKT_condition_lb} to derive a superior KKT stationary point by replacing ${\bf{A}}_{\sf{KKT}}(\bar{{\bf{f}}})$ and ${\bf{B}}_{\sf{KKT}}(\bar{{\bf{f}}})$ with ${\bf{A}}^{\sf lb}_{\sf{KKT}}(\bar{{\bf{f}}})$ and ${\bf{B}}^{\sf lb}_{\sf{KKT}}(\bar{{\bf{f}}})$, respectively.
We call it as robust SSSE-GPI-RSMA (RSSSE-GPI-RSMA) algorithm.

\subsection{Complexity Analysis}

The computational complexity of our proposed algorithm hinges on the calculation of the ${\bf{B}}^{-1}_{\sf KKT} (\bar {\bf{f}}_{t})$ matrix.
Given that the matrix size is $N(K+1) \times N(K+1)$, the typical complexity for inverting ${\bf{B}}_{\sf KKT} (\bar {\bf{f}}_{t})$ is on the order of $\CMcal{O}\left((K+1)^3N^3\right)$.
However, leveraging the block diagonal structure of ${\bf{B}}_{\sf KKT}(\bar {\bf{f}}_{t})$, it becomes feasible to reduce the computational complexity to $\CMcal{O}\left((K+1)N^3\right)$.
This optimization arises by computing the inversion of each $N \times N$ sub-matrix separately, given that there are $(K+1)$ such sub-matrices.
Let us denote the variable $T_{\sf GPI}$, $T_{\sf WMMSE}$, and $T_{\sf SCA}$ as the representative count of total iterations for each algorithm.
Then, the total complexity order of the proposed algorithm is given as $\CMcal{O}\left(T_{\sf GPI}(K+1)N^3\right)$.
Similarly,  the complexity of the convex relaxation method proposed in \cite{joudeh2016sum}, which relies on a quadratically constrained quadratic program, can be expressed in terms of the order of $\CMcal{O}\left(T_{\sf WMMSE}K^{3.5}N^{3.5}\right)$.
In \cite{xia2022secure}, the joint WMMSE and SCA based algorithm is classified as a specialized SCA method, characterized by a complexity order of $\CMcal{O}\left(T_{\sf SCA}K^{3.5}N^4\right)$.
In this regard, the proposed algorithm demonstrates a more computationally efficient approach than the state-of-the-art precoding methods.

\section{Numerical Results}

In this section, we conduct comparison of the sum secrecy SEs achieved by the proposed algorithm and existing baseline methods.
We adopt the one-ring model as proposed for the spatial channel covariance matrix \cite{adhikary2013joint}.
The user channel covariance between the AP's antenna $a$ and $b$ is defined as
\begin{align}
    \label{eq:user_channel_covariance}
    \left[{\bf{R}}_k \right]_{a,b} = \frac{1}{2\Delta_k} \int_{\theta_k - \Delta_k}^{\theta_k + \Delta_k} e^{-j \frac{2\pi}{\psi} \Psi(x) ({\bf{r}}_a - {\bf{r}}_b)} {\rm d} x,
\end{align}
where $\left[{\bf{R}}_k \right]_{a,b}$ denotes the $(a, b)$th element of ${\bf{R}}_k$,  $\Delta_k$ is the angular spread of the user $k$, $\theta_k$ is angle-of-arrival (AoA) of user $k$, $\Psi(x) = \left[\cos(x), \sin(x) \right]$, and ${\bf{r}}_a$ is the position vector of the $a$th antenna.
With the Karhunen-Loeve model,
the channel vector ${\bf{h}}_k$ can be decomposed as  \cite{adhikary2013joint} 
${\bf{h}}_k = {\bf{U}}_k {\bf \Lambda}_k^{\frac{1}{2}} {\pmb{\zeta}}_k,$
where ${\bf \Lambda}_k \in \mathbb{C}^{r_k \times r_k}$ is a diagonal matrix that contains the non-zero eigenvalues of the channel covariance matrix ${\bf{R}}_k$, ${\bf{U}}_k \in \mathbb{C}^{N \times r_k}$ is the corresponding eigenvectors,
and $ {\pmb{\zeta}}_k \in \mathbb{C}^{r_k}$ is a independent and identically distributed complex Gaussian vector with zero mean and unit variance.
Similar to the treatment of user channels, the eavesdroppers' channels adhere to analogous principles in their modeling and characterization.

For CSIT acquisition, we consider frequency division duplex (FDD) systems  whose estimated channel is given as \cite{wagner2012large}
\begin{align}
    \hat {\bf{h}}_k &= {\bf{U}}_k {\bf \Lambda}_k^{\frac{1}{2}} \left(\sqrt{1 - \kappa^2}  {\pmb{\zeta}}_k + \kappa {\bf{v}}_k \right) = \sqrt{1 - \kappa^2} {\bf{h}}_k + {\bf{q}}_k,
\end{align}
where $\hat {\bf{h}}_k$ is the estimated channel vector of user $k$, ${\bf{v}}_k$ is drawn from IID $ \mathcal{CN}(0,1)$, ${\bf{q}}_k$ is the CSIT quantization error vector, and $\kappa$ is a parameter that determines the quality of the channel with $0 \le \kappa \le 1$. 
With \eqref{eq:estimated_ch}, the error covariance is derived as 
\begin{align}
    \mathbb{E}[{\pmb{\phi}}_{k} {\pmb{\phi}}_k^{\sf H}] = {\bf{\Phi}}_k = {\bf{U}}_k {\bf \Lambda}_k^{\frac{1}{2}} (2 - 2\sqrt{1 - \kappa^2}) {\bf \Lambda}_k^{\frac{1}{2}} {\bf{U}}_k^{\sf H}.
\end{align}

\subsection{Performance Evaluation}
To evaluate the effectiveness of the proposed algorithms, SSSE-GPI-RSMA and RSSSE-GPI-RSMA, we compare its performance with the baseline methods.
The baselines for comparison are as follows:
\begin{itemize}
    \item {\bf{SSE-SCA-RSMA}}: This method is the joint WMMSE and SCA-based RSMA sum secrecy SE (SSE) maximization algorithm in \cite{xia2022secure}.
    Although this algorithm treats other users as potential eavesdroppers, it does not account for the presence of actual eavesdroppers.
To facilitate comparison with the our proposed algorithm, we establish a minimum user rate threshold, $R_{\sf{s}}^{\sf{th}}$, assigning it a sufficiently small value and ignore the user weights.
    \item {\bf{SE-WMMSE-RSMA}}: WMMSE-based RSMA algorithm for sum SE maximization without security consideration proposed in \cite{joudeh2016sum} under perfect CSIT and partial CSIT.
\end{itemize}

For the initialization, we use the precoding vector ${\bf{f}}_k = {\hat {\bf{h}}}_k$ employing maximum ratio transmission (MRT) strategy.
Specifically for the common stream, we utilize the average of the channel vectors for initialization.
We set $\Delta_k = \pi/6, \kappa = 0.4, \sigma^2 = 1, \epsilon = 0.05$, and  $t_{\rm max} = 100$, and  $\sigma^2 = \sigma^2_{\sf e}$. 
We also set $\alpha_1 = \alpha_2$ which is tuned empirically.
Unless otherwise stated, it is assumed that the users are distributed within a sector of $\pi/6$ radians and eavesdroppers are randomly positioned.
In other words, for the correlated channel, the angular difference between users, $\theta_k - \theta_{k'}$, with $\theta_k$ defined in \eqref{eq:user_channel_covariance}, is $\pi/6$.

\begin{figure}[!t]\centering
    \begin{subfigure}
       [Perfect CSIT case]{\resizebox{\columnwidth}{!}{\includegraphics{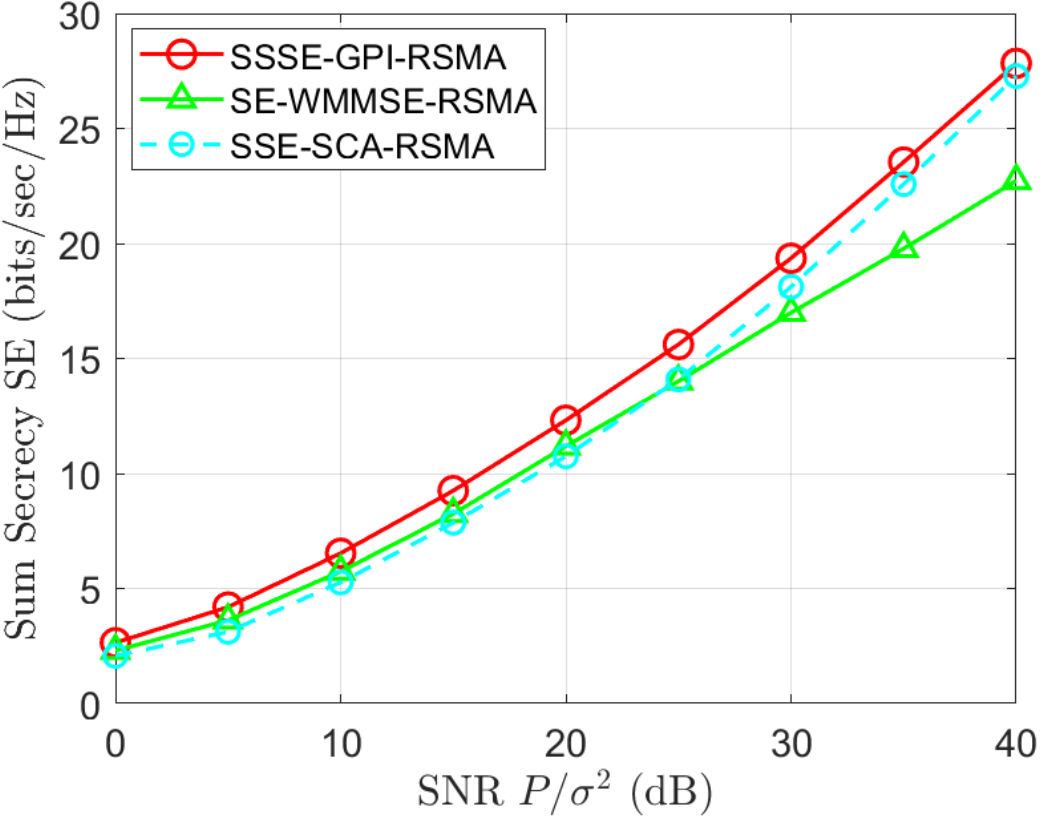}}}
    \end{subfigure}
    \begin{subfigure}
        [Partial CSIT case]{\resizebox{\columnwidth}{!}{\includegraphics{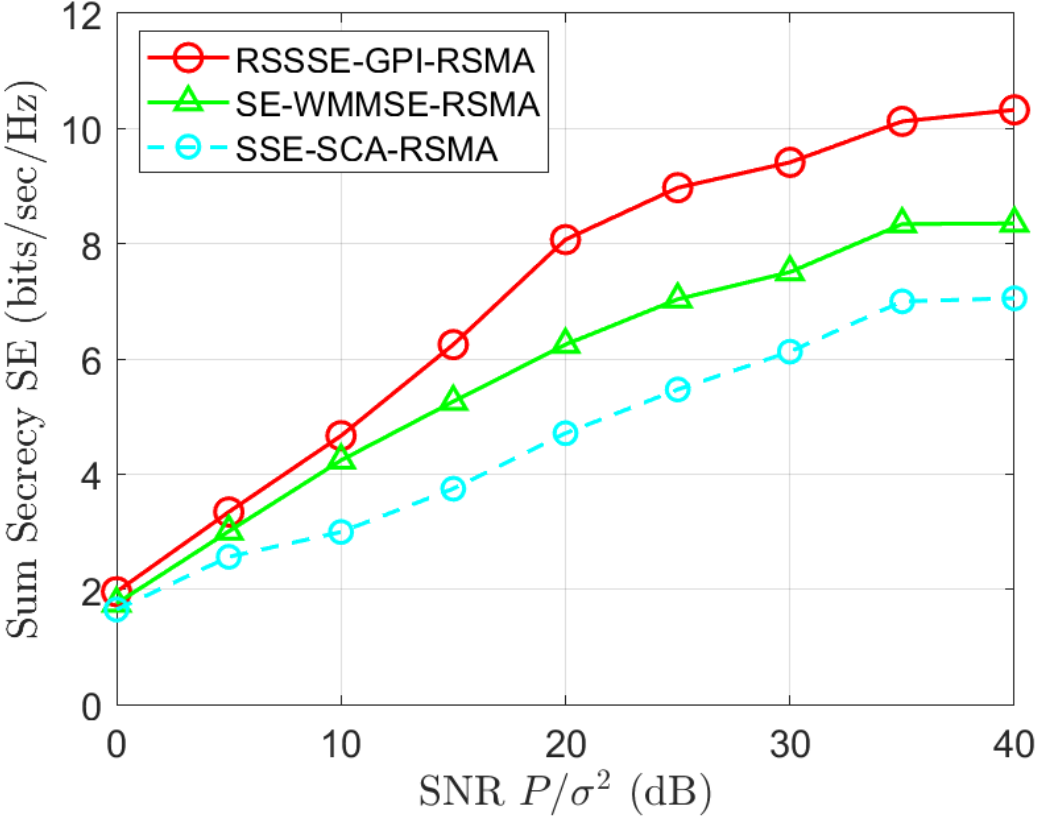}}}
    \end{subfigure}
    
    \caption{
    The sum secrecy SE versus SNR for $N = 4$ AP antennas, $K = 4$ total users, $S = 2$ secret users, $M = 2$ normal users, and $E = 2$ eavesdroppers.}
    \label{fig:SumSecrecySEversusSNR}
\end{figure}

\subsection{Secrecy SE versus SNR}

Fig.~\ref{fig:SumSecrecySEversusSNR} illustrates the sum secrecy SE as a function of the signal-to-noise ratio (SNR), defined by $P/\sigma^2$, in a general scenario where multiple users—comprising both secret and normal users—coexist alongside eavesdroppers.
In Fig.~\ref{fig:SumSecrecySEversusSNR}(a), the performance under perfect CSIT is analyzed.
Here, the proposed SSSE-GPI-RSMA algorithm emerges as the top performer, showcasing the highest sum secrecy SE.
The SSE-SCA-RSMA algorithm demonstrates comparable efficacy, closely trailing the performance of the proposed algorithm.
The performance disparity between SSSE-GPI-RSMA and SSE-SCA-RSMA can be attributed to the fundamental difference in security optimization approaches.
While SSSE-GPI-RSMA employs targeted security optimization for only secret users, SSE-SCA-RSMA implements uniform security measures across all users.
This strategic distinction demonstrates the superior performance of SSSE-GPI-RSMA, particularly in heterogeneous user environments across diverse SNR regimes.
The performance gap between SE-WMMSE-RSMA and its counterparts becomes more pronounced in the high SNR regime.
This degradation stems from inherent limitation of SE-WMMSE-RSMA in addressing information leakage optimization.
As the SNR increases, the impact of inter-user information leakage becomes more substantial, revealing SE-WMMSE-RSMA's inability to effectively mitigate this interference.

Similarly, Fig.~\ref{fig:SumSecrecySEversusSNR}(b) presents a performance of the sum secrecy SE under partial CSIT, where the proposed RSSSE-GPI-RSMA algorithm demonstrates superior performance.
Notably, in contrast to Fig.~\ref{fig:SumSecrecySEversusSNR}(a), SE-WMMSE-RSMA outperforms SSE-SCA-RSMA in this scenario.
This performance difference can be attributed to the stringent security constraints inherent in SSE-SCA-RSMA, highlighting the fundamental trade-off between enhanced security measures and system efficiency.
The comprehensive evaluation presented in Fig.~\ref{fig:SumSecrecySEversusSNR} validates the robustness of our proposed algorithm across diverse operational scenarios, maintaining consistent performance in both perfect and imperfect CSIT conditions with heterogeneous user environments.

\begin{figure}[!t]\centering
   \begin{subfigure}
   [$M=4$ normal users only]{\resizebox{\columnwidth}{!}{\includegraphics{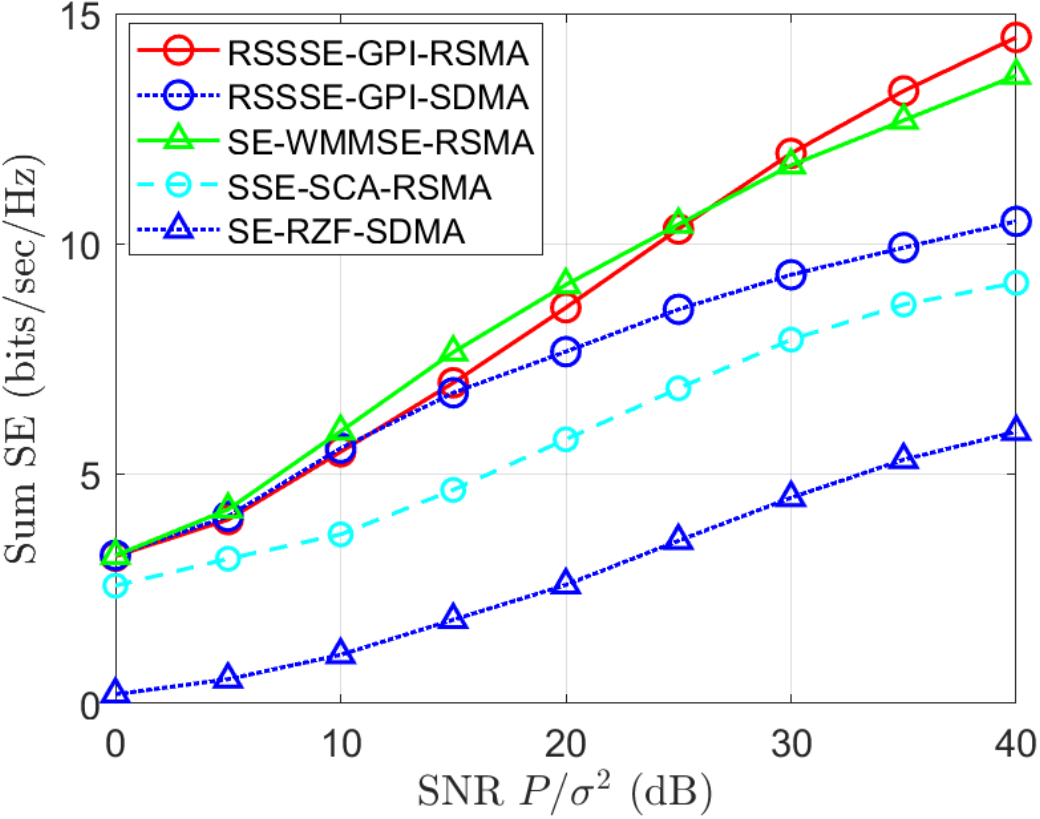}}}
    \end{subfigure}
    \begin{subfigure}
    [$S=4$ secret users only]{\resizebox{\columnwidth}{!}{\includegraphics{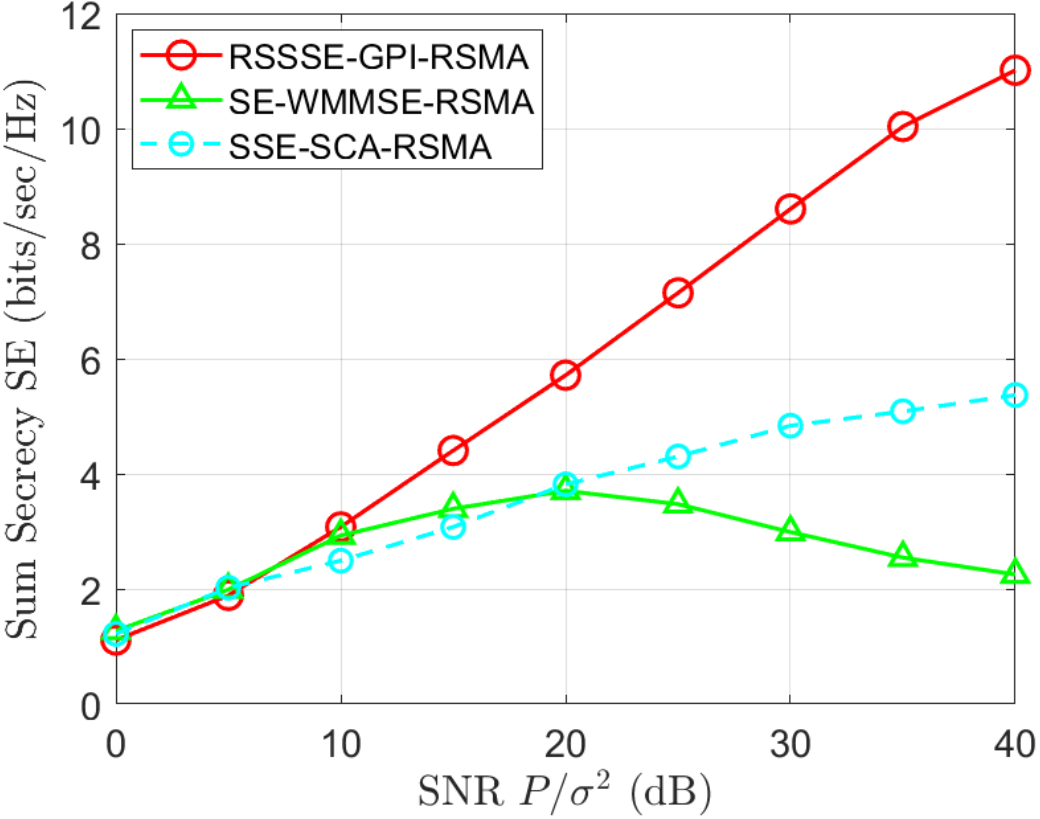}}}
    \end{subfigure}
    \caption{The sum secrecy SE versus SNR for $N = 4$ AP antennas, $K=4$ users, $E = 0$ eavesdroppers under partial CSIT.} 
    \label{fig:SumSecrecySEversusMandS}
\end{figure}

Fig.~\ref{fig:SumSecrecySEversusMandS} displays the variations in sum secrecy SE with respect to SNR under specialized scenarios: exclusively composed of normal users and secret users, respectively.
In Fig.~\ref{fig:SumSecrecySEversusMandS}(a), SDMA-based algorithms are introduced to facilitate a fair comparative analysis between RSMA and SDMA approaches.
This methodological choice stems from the fundamental architectural assumption: RSMA applies security objectives solely to the private stream, excluding the common component.
For performance metrics independent of security considerations, both algorithms provide a valid basis for comparison.
The following SDMA-based algorithms are implemented in this analysis:
\begin{itemize}
\item {\bf{RSSSE-GPI-SDMA}}: The proposed GPI-based algorithm for sum secrecy SE maximization without a common stream of proposed RSSSE-GPI-RSMA algorithm.
\item {\bf{SE-RZF-SDMA}}: The conventional linear regularized zero-forcing (RZF) precoder for sum SE maximization.
\end{itemize}
In Fig.~\ref{fig:SumSecrecySEversusMandS}(a), both the proposed RSSSE-GPI-RSMA and SE-WMMSE-RSMA algorithms exhibit the highest sum SE performance.
This comparable performance pattern emerges in scenarios with normal users only, where stringent security optimization is not essential.
While SSE-SCA-RSMA is designed to minimize information leakage to potential eavesdroppers, such robust security measures prove redundant in normal-user environments.
Consequently, the SSE-SCA-RSMA algorithm experiences performance degradation attributed to its rigorous security constraints.
This effect becomes increasingly pronounced as the SNR increases, leading to a widening performance gap between the SSE-SCA-RSMA and the SE-WMMSE-RSMA algorithms.
The RSMA advantage is evident in the performance differential between RSSSE-GPI-RSMA and RSSSE-GPI-SDMA.
In this figure, we can observe the importance of balanced security protocols that optimize system efficiency.

Contrasting with the scenario in Fig.~\ref{fig:SumSecrecySEversusMandS}(a), Fig.~\ref{fig:SumSecrecySEversusMandS}(b) demonstrates the impact of security optimization for secret users.
The SE-WMMSE-RSMA algorithm exhibits notable performance degradation beyond 20 dB, primarily due to increased inter-user signal leakage at higher SNR values.
This performance characteristic emphasizes the critical need for security optimization in scenarios exclusively comprising secret users to ensure secure communication.
In other words, the performance gap observed between the SSE-SCA-RSMA and SE-WMMSE-RSMA algorithms serves as a quantifiable measure of the gains achieved through security optimization.
The comprehensive analysis in Fig.~\ref{fig:SumSecrecySEversusMandS} validates our algorithm's adaptability through user categorization, effectively addressing varying security requirements across different scenarios.

\begin{figure}[!t]\centering
   \subfigure{\resizebox{\columnwidth}{!}{\includegraphics{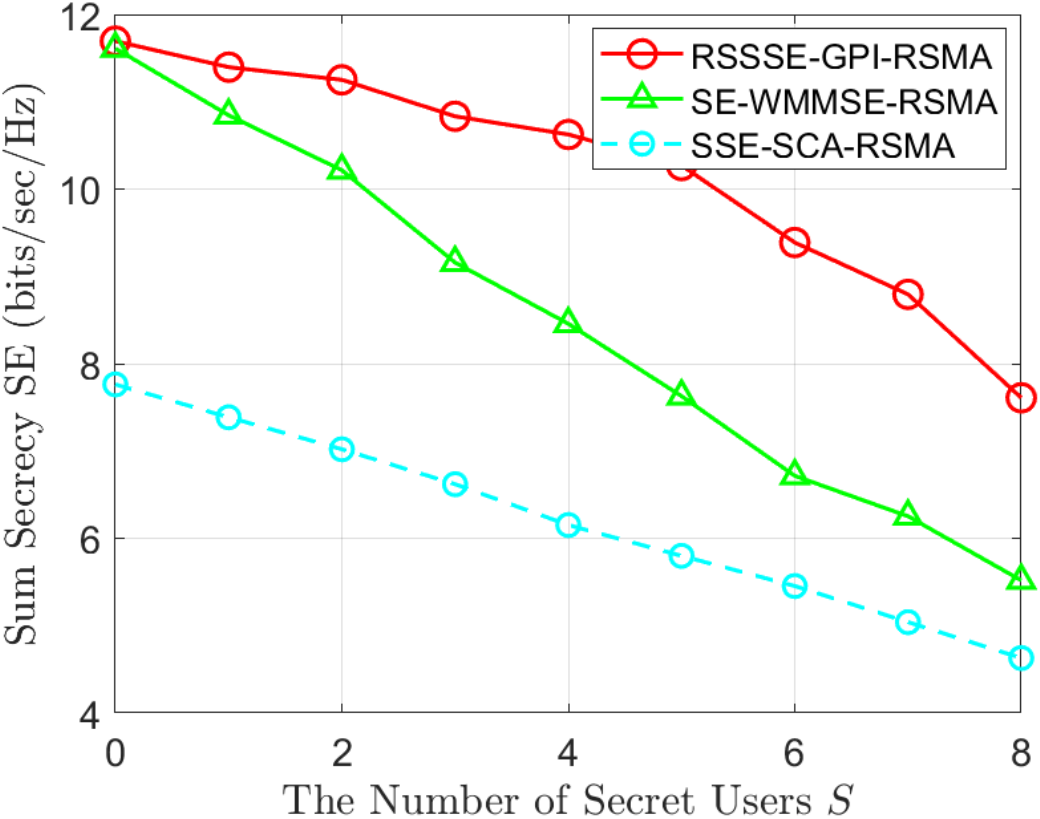}}}
       \vspace{-0.5em}
   \caption{The sum secrecy SE versus the number of secret users $S$ for $N = 8$ AP antennas, $K = 8$ total users, $E = 0$ eavesdroppers, and SNR = 10 dB under partial CSIT.}
   \label{fig:SumSecrecySEversusSU_partial}
    \vspace{-1em}
\end{figure}

\subsection{Secrecy SE versus Number of Users and Eavesdroppers}
Fig.~\ref{fig:SumSecrecySEversusSU_partial} illustrates the impact of varying the number of secret users on sum secrecy SE performance.
The increasing number of secret users necessitates enhanced security optimization, leading to performance degradation across all evaluated algorithms due to greater security requirements.
Despite this degradation, the proposed algorithm maintains superior performance compared to baseline approaches.
The SE-WMMSE-RSMA algorithm exhibits the most significant performance decline, attributable to its design focus on sum SE maximization without explicit security considerations.
This architectural limitation results in increased vulnerability to information leakage as the number of secret users grows, demonstrating the critical interplay between SE optimization and security requirements in wireless networks.

\begin{figure}[!t]\centering
   \subfigure{\resizebox{\columnwidth}{!}{\includegraphics{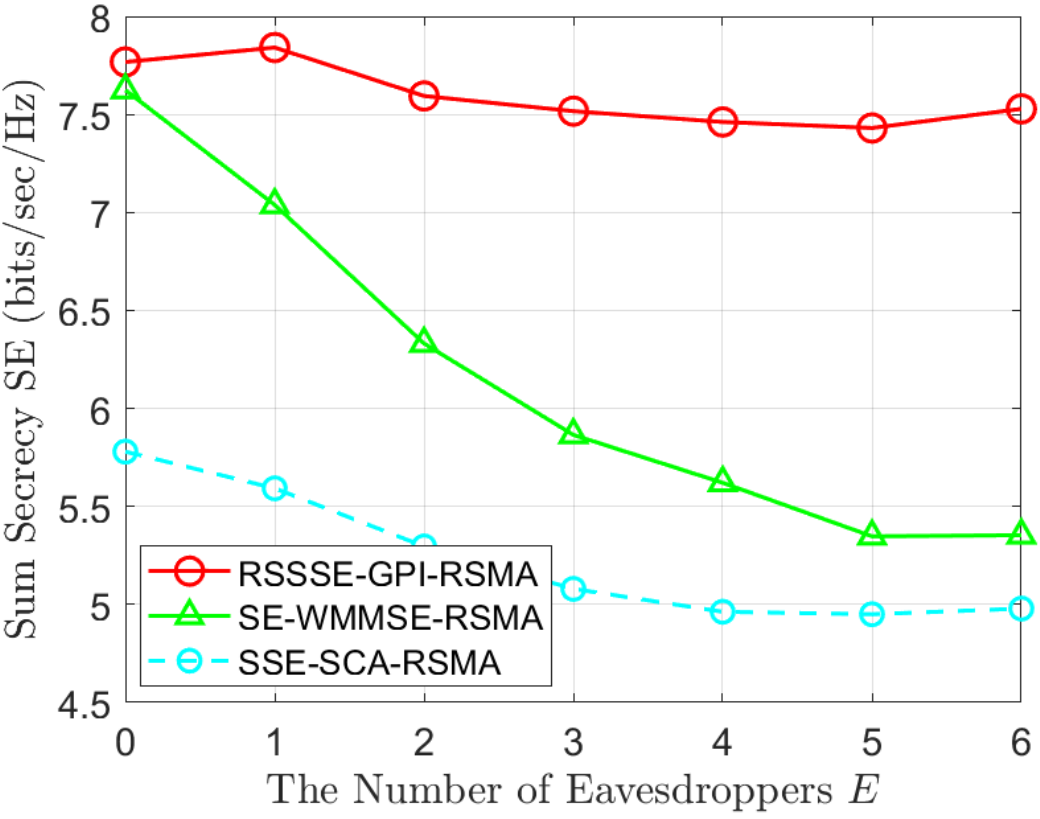}}}
       \vspace{-0.5em}
   \caption{The sum secrecy SE versus the number of eavesdroppers for $N = 4$ AP antennas, $S = 4$ secret users, $M = 0$ normal users and SNR = 20 dB under partial CSIT.}
   \label{fig:SumSecrecySEversusEve_partial}
\end{figure}

Fig.~\ref{fig:SumSecrecySEversusEve_partial} examines the impact of increasing the number of eavesdroppers on sum secrecy SE at 20 dB SNR, with randomly distributed legitimate users and adjacent eavesdroppers.
The proposed algorithm demonstrates superior performance compared to baseline methods. 
Although all algorithms exhibit performance degradation with increasing the number of eavesdroppers, the proposed RSSSE-GPI-RSMA algorithm reveals the robust performance by  explicitly accounting for  eavesdroppers in the precoder design.
Conversely, SSE-SCA-RSMA shows gradual performance degradation due to its limitation of considering only legitimate users as potential eavesdroppers, leading to increased information leakage.
SE-WMMSE-RSMA exhibits the most significant performance decline due to its lack of eavesdropper consideration.
These results demonstrate the importance of incorporating comprehensive security mechanisms in RSMA algorithms for maintaining secrecy SE in eavesdropper-rich environments.

\begin{figure}[!t]\centering
   \subfigure{\resizebox{\columnwidth}{!}{\includegraphics{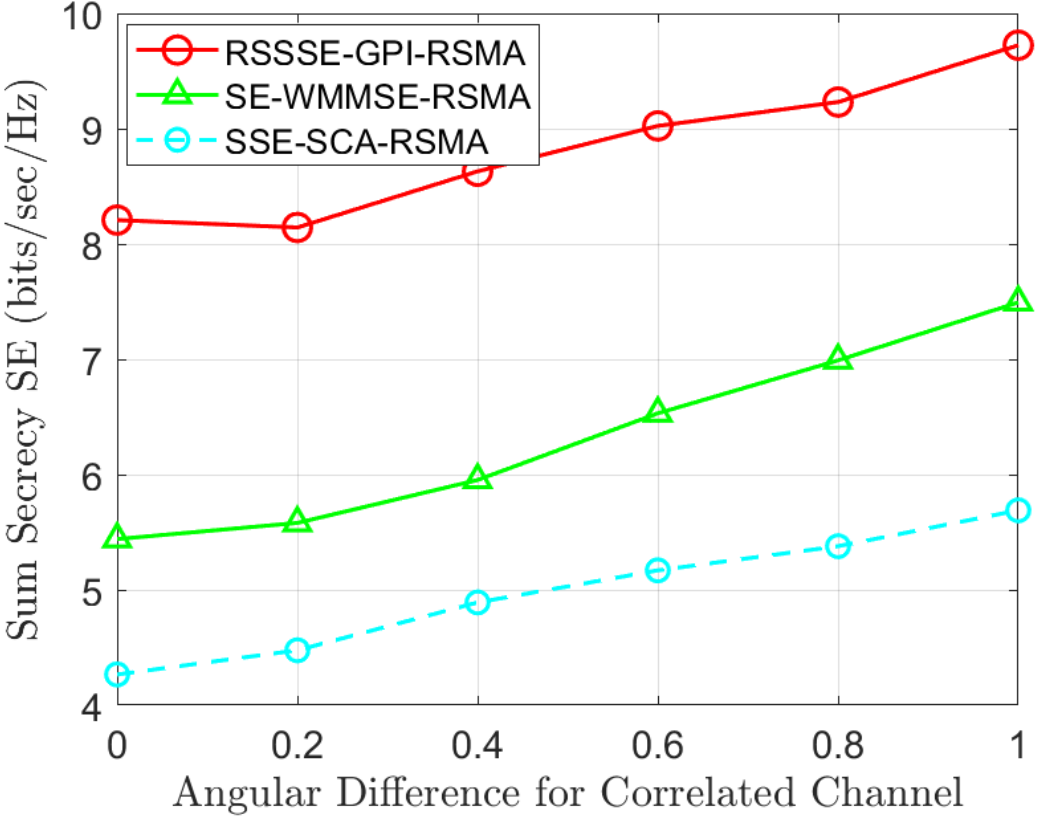}}}
       \vspace{-0.5em}
   \caption{The sum secrecy SE versus the angular difference for correlated channel $\theta_k - \theta_{k'}$ for $N = 6$ AP antennas, $S = 4$ secret users, $M = 2$ normal users, $E = 2$ eavesdroppers and SNR = 20 dB under partial CSIT.} \label{fig:SumSecrecySEversusCorrelation_partial}
    \vspace{-1em}
\end{figure}

\subsection{Secrecy SE versus Channel Condition}
Fig.~\ref{fig:SumSecrecySEversusCorrelation_partial} shows the change in sum secrecy SE according to $\theta_k - \theta_{k'}$, the angular difference parameter between users.
All algorithms exhibit improved performance with increasing angular difference due to reduced inter-user interference.
The proposed RSSSE-GPI-RSMA algorithm consistently shows the highest performance, significantly surpassing both SE-WMMSE-RSMA and SSE-SCA-RSMA.
SSE-SCA-RSMA shows lower performance due to its stringent security constraints that considers all RSMA users as pontential eavesdroppers.
These results validate the effectiveness of our optimization approach across varying angular difference conditions.

\begin{figure}[!t]\centering
   \subfigure{\resizebox{\columnwidth}{!}{\includegraphics{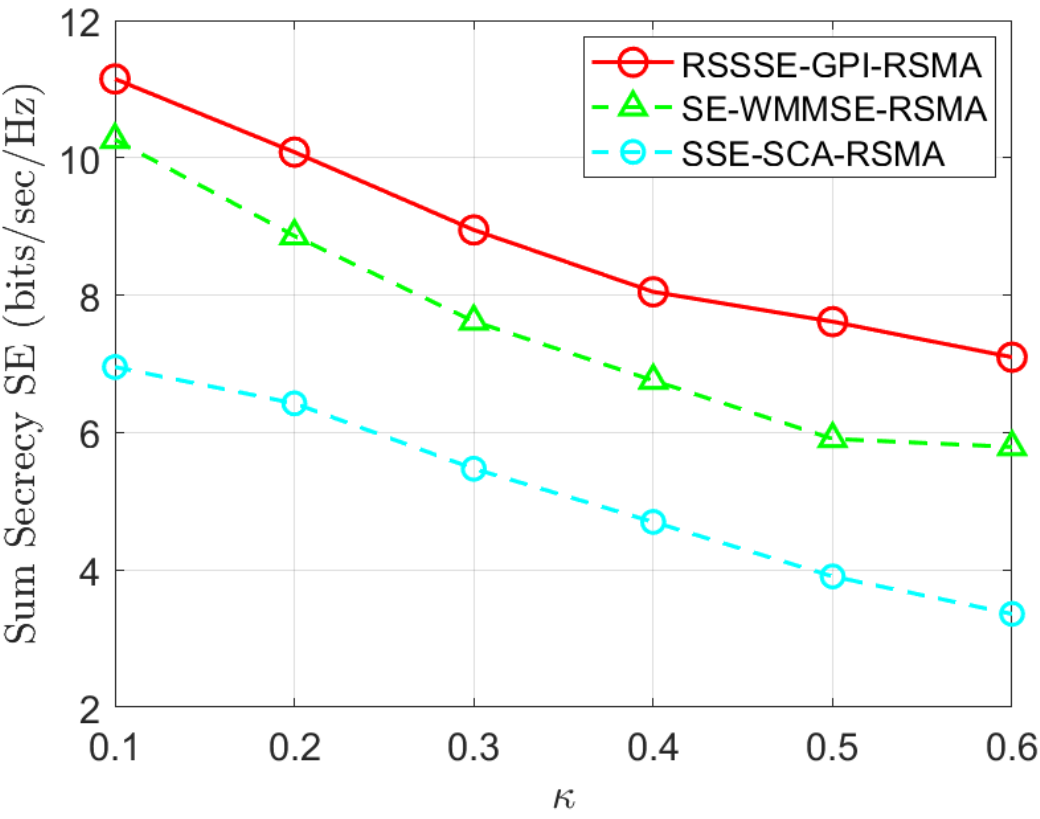}}}
       \vspace{-0.5em}
   \caption{The sum secrecy SE versus the quality of the channel $\kappa$ for $N = 4$ AP antennas, $S = 2$ secret users, $M = 2$ normal users, $E = 2$ eavesdroppers and SNR = 20 dB under partial CSIT.} \label{fig:SumSecrecySEversusKappa_partial}
    \vspace{-1em}
\end{figure}

Fig.~\ref{fig:SumSecrecySEversusKappa_partial} depicts the relationship between sum secrecy SE and channel quality parameter $\kappa$.
All algorithms exhibit performance degradation with increasing $\kappa$, indicating sensitivity to deteriorating channel conditions.
Notably, the proposed algorithm consistently outperforms the others across all sections, highlighting its superior capability.
Fig.~\ref{fig:SumSecrecySEversusKappa_partial} validates the robustness of our proposed algorithm in practical scenarios with imperfect channel conditions.

\begin{figure}[!t]\centering
   \subfigure{\resizebox{\columnwidth}{!}{\includegraphics{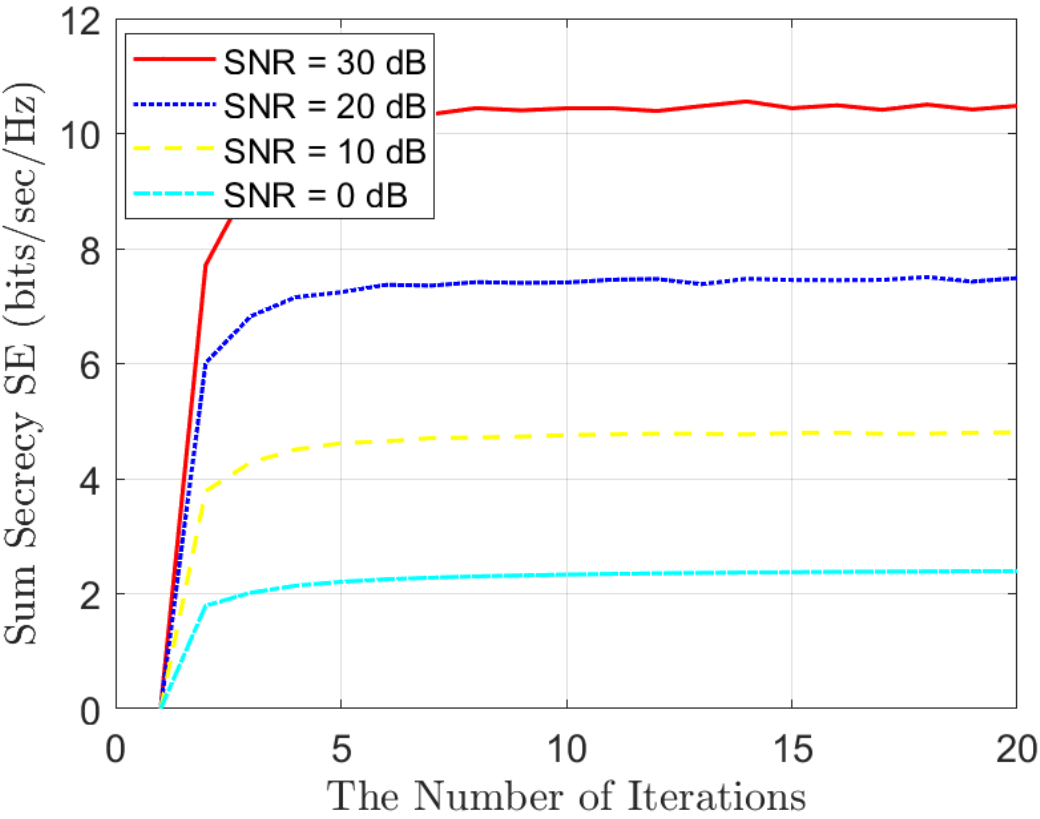}}}
       \vspace{-0.5em}
   \caption{Convergence results in terms of the sum secrecy SE for $N = 4$ AP antennas, $S = 2$ secret users, $M = 2$ normal users, and $E = 2$ eavesdroppers under partial CSIT.} \label{fig:SumSecrecySEversusConvergence_partial}
    \vspace{-1em}
\end{figure}

\subsection{Convergence}
Fig.~\ref{fig:SumSecrecySEversusConvergence_partial} displays the convergence outcomes relative to the number of iterations for SNR values in the set $\left\{0,10,20,30\right\}$ dB.
The proposed RSSSE-GPI-RSMA algorithm achieves fast convergence within 7 iterations across all the specified SNR levels, which is relatively quick. 
Consequently, the proposed algorithm is more advantageous for practical applications compared to other advanced precoding techniques, as it presents significantly reduced complexity as discussed and achieves fast convergence.

\section{Conclusion}
In this study, we introduced a precoding algorithm that integrates RSMA with security consideration, initially under perfect CSIT, subsequently extending its applicability to scenarios with partial CSIT.
Given the inherent non-smooth and non-convex nature of the problem under both CSIT conditions, we employed the LogSumExp technique for problem smoothing, further reformulating the problem into a variant of the Rayleigh quotient to enhance tractability.
For scenarios characterized by limited CSIT, it became imperative to leverage the conditional average SE, employing lower bounding and approximation strategies for practicality.
By deriving the stationary condition of the reformulated problem and mapping it onto a nonlinear eigenvalue problem, we were able to pinpoint the most advantageous local optimal solution.
Our simulation results demonstrated the superiority of the proposed method over existing state-of-the-art algorithms, particularly in achieving the highest sum secrecy SE.
The primary advantage of our proposed algorithm is its ability to accommodate multiple users with heterogeneous privacy priorities and eavesdroppers within a general framework under limited CSIT. 
Consequently, our algorithm ensures higher sum secrecy SE compared to existing baselines, significantly advancing secure wireless communications.

\appendices
\section{Proof of Lemma \ref{lem:KKT}}
\label{app:lemma1}
It is assumed that we employ the maximum power of the precoder $\|\bar\bff\|^2=1$.
Since the problem in \eqref{eq:op_Rayleigh} with the reformulated objective function is invariant up to the scaling of $\bar\bff$, we can indeed disregard the power constraint in this context.
Let the objective function in \eqref{eq:op_Rayleigh} as $L(\bar {\bf f})$. 
Then $L(\bar {\bf f})$ itself is a Lagrangian function.
We consider $L(\bar {\bf f}) = L_1(\bar{{\bf{f}}}) + L_2(\bar{{\bf{f}}}) + L_3(\bar{{\bf{f}}}) + L_4(\bar{{\bf{f}}})$
where
\begin{align}
    \nonumber
    &L_1(\bar{{\bf{f}}}) = \ln\left(\sum_{k \in \CMcal{K}} \left(\frac{\bar{\bf{f}}^{\sf{H}} {\bf{A}}_{{\sf{c}},k}\bar{\bf{f}}}{\bar{\bf{f}}^{\sf{H}} {\bf{B}}_{{\sf{c}},k}\bar{\bf{f}}}\right)^{-\frac{1}{{\alpha}\ln2}}\right)^{-\alpha}, \\
    \nonumber
    &L_2(\bar{{\bf{f}}}) = \sum_{s\in {\CMcal{S}}}\log_2 \left( \frac{\bar{\bf{f}}^{\sf{H}} {\bf{A}}_{s} \bar{\bf{f}}}{\bar{\bf{f}}^{\sf{H}} {\bf{B}}_{s}\bar{\bf{f}}} \right), \\
    \nonumber
    &L_3(\bar{{\bf{f}}}) = -\sum_{s\in {\CMcal{S}}}\ln\left(\sum_{{\sf{e}} \in \CMcal{E}} \left( \frac{\bar{\bf{f}}^{\sf{H}} {\bf{A}}_{\sf{e}}^{(s)}\bar{\bf{f}}}{\bar{\bf{f}}^{\sf{H}} {\bf{B}}_{\sf{e}}^{(s)}\bar{\bf{f}}} \right)^{\frac{1}{\alpha\ln2}}\! \!+\! \!\sum_{u \in \CMcal{K}\setminus\{s\}} \left( \frac{\bar{\bf{f}}^{\sf{H}} {\bf{C}}_u^{(s)}\bar{\bf{f}}}{\bar{\bf{f}}^{\sf{H}} {\bf{D}}_u^{(s)}\bar{\bf{f}}} \right)^{\frac{1}{\alpha\ln2}}\right)^\alpha,
    \\\nonumber
    &L_4(\bar{{\bf{f}}}) = \sum_{m\in {\CMcal{M}}} \log_2 \left( \frac{\bar{\bf{f}}^{\sf{H}} {\bf{A}}_{m} \bar{\bf{f}}}{\bar{\bf{f}}^{\sf{H}} {\bf{B}}_{m}\bar{\bf{f}}} \right).
\end{align}
To derive the first-order Karush-Kuhn-Tucker (KKT) condition, we solve $\frac{\partial L(\bar{{\bf{f}}})}{\partial \bar{{\bf{f}}}^{\sf{H}}}=0$.
Utilizing the following derivative:
\begin{align}
    \partial\left(\frac{{\bar{{\bf{f}}}}^{\sf{H}}{{\bf{A}}}{\bar{{\bf{f}}}}}{{\bar{{\bf{f}}}}^{\sf{H}}{{\bf{B}}}{\bar{{\bf{f}}}}}\right)/\partial{\bar{{\bf{f}}}}^{\sf{H}} = 2\times\frac{{\bar{{\bf{f}}}}^{\sf{H}}{\bf{A}}{\bar{{\bf{f}}}}}{{\bar{{\bf{f}}}}^{\sf{H}}{\bf{B}}{\bar{{\bf{f}}}}}\left[\frac{{\bf{A}}{\bar{{\bf{f}}}}}{{\bar{{\bf{f}}}}^{\sf{H}}{\bf{A}}{\bar{{\bf{f}}}}}-\frac{{\bf{B}}{\bar{{\bf{f}}}}}{{\bar{{\bf{f}}}}^{\sf{H}}{\bf{B}}{\bar{{\bf{f}}}}}\right],
\end{align}
the partial derivative of each component can be expressed as
\begin{align}
    \nonumber
    \frac{\partial L_1(\bar{{\bf{f}}})}{\partial \bar{{\bf{f}}}^{\sf{H}}} &= \frac{2}{\ln2} \sum_{k \in \CMcal{K}}\left[\frac{w_{{\sf{c}},k}(\bar{\bf{f}})}{\sum_{l \in \CMcal{K}}w_{{\sf{c}},l}(\bar{\bf{f}})}\left\{\frac{{{\bf{A}}_{{\sf{c}},k}}{\bar{{\bf{f}}}}}{{\bar{{\bf{f}}}}^{\sf{H}}{\bf{A}}_{{\sf{c}},k}{\bar{{\bf{f}}}}}-\frac{{\bf{B}}_{{\sf{c}},k}{\bar{{\bf{f}}}}}{{\bar{{\bf{f}}}}^{\sf{H}}{\bf{B}}_{{\sf{c}},k}{\bar{{\bf{f}}}}}\right\}\right], \\
    \nonumber
    \frac{\partial L_2(\bar{{\bf{f}}})}{\partial \bar{{\bf{f}}}^{\sf{H}}} &= \frac{2}{\ln2} \sum_{s\in \CMcal{S}} \left[\frac{{\bf{A}}_s{\bar{{\bf{f}}}}}{{\bar{{\bf{f}}}}^{\sf{H}}{\bf{A}}_s{\bar{{\bf{f}}}}}-\frac{{\bf{B}}_s{\bar{{\bf{f}}}}}{{\bar{{\bf{f}}}}^{\sf{H}}{\bf{B}}_s{\bar{{\bf{f}}}}}\right], 
    \\
    \nonumber
    \frac{\partial L_3(\bar{{\bf{f}}})}{\partial \bar{{\bf{f}}}^{\sf{H}}} &= -\frac{2}{\ln2}\sum_{s\in \CMcal{S}} \left[\frac{\sum_{{\sf{e}} \in \CMcal{E}}w_{{\sf{e}}}^{(s)}(\bar{\bf{f}})\left\{\frac{{{\bf{A}}_{\sf{e}}^{(s)}}{\bar{{\bf{f}}}}}{{\bar{{\bf{f}}}}^{\sf{H}}{\bf{A}}_{\sf{e}}^{(s)}{\bar{{\bf{f}}}}}-\frac{{\bf{B}}_{\sf{e}}^{(s)}{\bar{{\bf{f}}}}}{{\bar{{\bf{f}}}}^{\sf{H}}{\bf{B}}_{\sf{e}}^{(s)}{\bar{{\bf{f}}}}}\right\}}{\sum_{i \in \CMcal{E}}w_{{\sf{i}}}^{(s)}(\bar{\bf{f}}) + \sum_{i \in \CMcal{K}\setminus\{s\}}w_{i}^{(s)}(\bar{\bf{f}})} \right. \\
    \nonumber
    &\quad + \left. \frac{\sum_{u \in \CMcal{K}\setminus\{s\}}w_{u}^{(s)}(\bar{\bf{f}})\left\{\frac{{{\bf{C}}_u^{(s)}}{\bar{{\bf{f}}}}}{{\bar{{\bf{f}}}}^{\sf{H}}{\bf{C}}_u^{(s)}{\bar{{\bf{f}}}}}-\frac{{\bf{D}}_u^{(s)}{\bar{{\bf{f}}}}}{{\bar{{\bf{f}}}}^{\sf{H}}{\bf{D}}_u^{(s)}{\bar{{\bf{f}}}}}\right\}}{\sum_{i \in \CMcal{E}}w_{{\sf{i}}}^{(s)}(\bar{\bf{f}}) + \sum_{i \in \CMcal{K}\setminus\{s\}}w_{i}^{(s)}(\bar{\bf{f}})} \right], \\
    \frac{\partial L_4(\bar{{\bf{f}}})}{\partial \bar{{\bf{f}}}^{\sf{H}}} &= \frac{2}{\ln2} \sum_{m\in {\CMcal{M}}} \left[\frac{{\bf{A}}_m{\bar{{\bf{f}}}}}{{\bar{{\bf{f}}}}^{\sf{H}}{\bf{A}}_m{\bar{{\bf{f}}}}}-\frac{{\bf{B}}_m{\bar{{\bf{f}}}}}{{\bar{{\bf{f}}}}^{\sf{H}}{\bf{B}}_m{\bar{{\bf{f}}}}}\right],
\end{align}
where $w_{{\sf{c}},k}(\bar{\bf{f}})$, $w_{{\sf{e}}}^{(s)}(\bar{\bf{f}})$, and $w_{u}^{(s)}(\bar{\bf{f}})$ are defined in Lemma~\ref{lem:KKT}.
The first-order KKT condition holds when
\begin{align} \label{eq:KKT}
    \frac{\partial L_1(\bar{{\bf{f}}})}{\partial \bar{{\bf{f}}}^{\sf{H}}} + \frac{\partial L_2(\bar{{\bf{f}}})}{\partial \bar{{\bf{f}}}^{\sf{H}}} + \frac{\partial L_3(\bar{{\bf{f}}})}{\partial \bar{{\bf{f}}}^{\sf{H}}} + \frac{\partial L_4(\bar{{\bf{f}}})}{\partial \bar{{\bf{f}}}^{\sf{H}}} = 0.
\end{align}
Carefully re-arranging the components in \eqref{eq:KKT}, the first-order KKT condition can be expressed equivalently as
\begin{align}
    {\bf{A}}_{\sf{KKT}}(\bar{{\bf{f}}})\bar{{\bf{f}}} = {\lambda}(\bar{{\bf{f}}}){\bf{B}}_{\sf{KKT}}(\bar{{\bf{f}}})\bar{{\bf{f}}},
\end{align}
where  ${\bf A}_{\sf KKT}$ and ${\bf B}_{\sf KKT}$ are defined in Lemma~\ref{lem:KKT}.
Given that ${\bf{B}}_{\sf{KKT}}(\bar{{\bf{f}}})$  exhibits Hermitian properties and maintains full rank, it is therefore invertible.
This concludes the proof.
\qed

\section{Proof of Lemma~\ref{sum_approx}}
\label{app:lemma2}
    We begin by considering the case for $N = 2$.
    Then,
    \begin{align}
        \mathbb{E} \left[\log_2 \left(1 + \sum_{i=1}^2\frac{X_i}{Y_i}\right)\right] &= \mathbb{E} \left[\log_2 \left(1 + \left\{\frac{X_1}{Y_1} + \frac{X_2}{Y_2}\right\}\right)\right] \\
        &= \mathbb{E} \left[\log_2 \left(1 + \left\{\frac{X_1Y_2 + X_2Y_1}{Y_1Y_2}\right\}\right)\right] \\
        &\stackrel{(a)}\approx \log_2 \left(1 + \frac{\mathbb{E}\left[X_1Y_2 + X_2Y_1\right]}{\mathbb{E}\left[Y_1Y_2\right]}\right) \\
        &\stackrel{(b)}= \log_2 \left(1 + \sum_{i=1}^2 \frac{\mathbb{E} \left[X_i\right]}{\mathbb{E} \left[Y_i\right]}\right),
    \end{align}
    where $(a)$ comes from Lemma 1 in \cite{zhang2014power}, and $(b)$ follows from the assumption of independence.
    It is straightforward to extend to $N>2$, yielding the  approximation result in \eqref{lemma_approx}.

\bibliographystyle{IEEEtran}
\bibliography{bibtext}

\end{document}